\documentclass{amsart}
\copyrightinfo{2017}{Nicholas Coxon}

\usepackage{amssymb,url}
\usepackage[foot]{amsaddr}

\newtheorem{theorem}{Theorem}
\newtheorem{lemma}[theorem]{Lemma}

\usepackage{algpseudocode,algorithm,caption}
\newcommand{\Statew}[2]{\State\parbox[t]{\dimexpr\linewidth-#1em}{#2}\strut}

\def\vec#1{\mathchoice{\mbox{\boldmath$\displaystyle#1$}}
{\mbox{\boldmath$\textstyle#1$}}
{\mbox{\boldmath$\scriptstyle#1$}}
{\mbox{\boldmath$\scriptscriptstyle#1$}}}
\newcommand{\floor}[1]{\left\lfloor #1\right\rfloor}
\newcommand{\ceil}[1]{\left\lceil #1\right\rceil}
\newcommand{\N}{\mathbb{N}}
\newcommand{\Z}{\mathbb{Z}}
\newcommand{\RR}{\mathbb{R}}
\newcommand{\F}{\mathbb{F}}
\newcommand{\I}{I}
\newcommand{\Ical}{\mathcal{I}}
\newcommand{\C}{C}
\newcommand{\R}{R}
\newcommand{\bigO}{\mathcal{O}}
\newcommand{\Mult}{\mathrm{Mult}}
\newcommand{\evH}{E}
\DeclareMathOperator{\eval}{ev}
\DeclareMathOperator{\mult}{\mathsf{M}}

\makeatletter
\newcommand*{\bdiv}{%
	\nonscript\mskip-\medmuskip\mkern5mu%
	\mathbin{\operator@font div}\penalty900\mkern5mu%
	\nonscript\mskip-\medmuskip
}
\makeatother

\let\originalleft\left
\let\originalright\right
\renewcommand{\left}{\mathopen{}\mathclose\bgroup\originalleft}
\renewcommand{\right}{\aftergroup\egroup\originalright}

\begin{document}

\title{Fast systematic encoding of multiplicity codes}
\author{Nicholas Coxon}
\address{\normalfont INRIA and Laboratoire d'Informatique de l'\'{E}cole polytechnique, Palaiseau, France.}
\email{nicholas.coxon@inria.fr}
\date{\today}

\begin{abstract}
	We present quasi-linear time systematic encoding algorithms for multiplicity codes. The algorithms have their origins in the fast multivariate interpolation and evaluation algorithms of van der Hoeven and Schost~(2013), which we generalise to address certain Hermite-type interpolation and evaluation problems. By providing fast encoding algorithms for multiplicity codes, we remove an obstruction on the road to the practical application of the private information retrieval protocol of Augot, Levy-dit-Vehel and Shikfa~(2014). 
\end{abstract}

\maketitle

\section{Introduction}

Multiplicity codes~\cite{kopparty2011,kopparty2014} generalise the classical family of Reed--Muller codes by augmenting their construction to include the evaluations of derivatives up to a given order. They inherit the property of being locally correctable from Reed--Muller codes, allowing any specified coordinate of a codeword in a multiplicity code to be recovered with high probability after examining only a sublinear, in the dimension of the code, number of entries in a possibly corrupted version of the codeword. Restricting to Reed--Muller codes while retaining sublinear local correction also restricts the maximum attainable information rate of the codes to roughly a half. Moving to multiplicity codes allows sublinear local correction and rates approaching one~\cite{kopparty2011}.

A closely related notion to local correction is that of local decoding~\cite{katz2000}. Whereas local correctability is a property of the codewords of a code, local decodability is a property of an encoding function of a code. For local decoding, one is required to recover a specified coordinate of a message after examining only a small number of coordinates in a possibly corrupted version of its encoding. It follows that a locally correctable code that is equipped with a systematic encoding function, i.e., one that embeds messages into their encodings, is also locally decodable. Augot, Levy-dit-Vehel and Ng\^{o}~\cite{augot2015} provide a systematic encoding function for multiplicity codes by combining results of Kopparty~\cite{kopparty2012} and Key, McDonough and Mavron~\cite{key2006}. By using their encoding function, multiplicity codes offer sublinear local decoding, while still allowing high rates.

The local decoding algorithm of multiplicity codes is randomised, with the queries to a codeword appearing uniformly distributed over its entries when viewed individually. As a result, an information-theoretically secure private information retrieval protocol may be built upon multiplicity codes by using the construction of Katz and Trevisan~\cite[Section~4]{katz2000}. Private information retrieval~\cite{chor1998} allows a user to retrieve entries from an online database without revealing which entries are being retrieved to the database servers. Using multiplicity codes in the construction of Katz and Trevisan yields a protocol with low communication complexity, when compared to the trivial solution of downloading the entire database, since the amount of data transferred to recover a single database entry is roughly equal to amount of codeword data examined during one round of local decoding.

Augot, Levy-dit-Vehel and Shikfa~\cite{augot2014} exploit geometric properties of multiplicity codes to improve upon the protocol obtain by the Katz--Trevisan construction, with their protocol incurring a smaller storage overhead and requiring fewer database servers. The protocol begins by systematically encoding the database as a codeword in a multiplicity code. The codeword is then distributed amongst the database servers. It follows that the encoding increases the amount of stored data by a factor equal to the inverse of the information rate of the code. Thus, the protocol favours the use of multiplicity codes over Reed--Muller codes.

For the protocol of Augot, Levy-dit-Vehel and Shikfa to be realisable for large databases, it is necessary that the initial encoding may be performed efficiently. In this paper, we show that it is possible to perform the encoding in time that is quasi-linear in the number of field elements that appear in the codeword.

\subsection{Multiplicity codes}

Let $\F_q$ denote the finite field with $q$ elements. We enumerate the field as $\F_q=\{\alpha_0,\dotsc,\alpha_{q-1}\}$ and let $[q]=\{0,1,\dotsc,q-1\}$ denote its index set. Then the elements of $\F^n_q$ are identified with vectors in $[q]^n$ by defining $\vec{\alpha}_{\vec{j}}=(\alpha_{j_1},\dotsc,\alpha_{j_n})$ for $\vec{j}=(j_1,\dotsc,j_n)\in[q]^n$. The ring of polynomials over $\F_q$ in indeterminates $X_1,\dotsc,X_n$ is denoted by $\F_q[\vec{X}]=\F_q[X_1,\dotsc,X_n]$, and we define $\vec{X}^{\vec{i}}=X^{i_1}_1\dotsm X^{i_n}_n$ for $\vec{i}=(i_1,\dotsc,i_n)\in\N^n$.

A codeword of a multiplicity code is constructed by taking a polynomial in $\F_q[\vec{X}]$ and evaluating its Hasse derivatives up to a given order at all points in $\F^n_q$. The Hasse derivatives of a polynomial $F\in\F_q[\vec{X}]$ are given by the coefficients (in $\F_q[\vec{X}]$) of the shifted polynomial $F(\vec{X}+\vec{Z})\in\F_q[\vec{X}][\vec{Z}]=\F_q[\vec{X}][Z_1,\dotsc,Z_n]$ for algebraically independent indeterminates $Z_1,\dotsc,Z_n$ over $\F_q[\vec{X}]$. For $\vec{s}=(s_1,\dotsc,s_n)\in\N^n$, the coefficient of $\vec{Z}^{\vec{s}}=Z^{s_1}_1\dotsm Z^{s_n}_n$ in the shifted polynomial is called the $\vec{s}$th Hasse derivative of $F$, which we denote by $H(F,\vec{s})$. Accordingly, we have
\begin{equation*}
	F(\vec{X}+\vec{Z})
	=\sum_{\vec{s}\in\N^n}
	H(F,\vec{s})(\vec{X})\vec{Z}^{\vec{s}}.
\end{equation*}
We define the weight of a vector $\vec{i}\in\N^n$, denoted $\left|\vec{i}\right|$, to be the sum of its entries. Then the $\vec{s}$th Hasse derivative is said to have order $\left|\vec{s}\right|$.

The polynomials that have their derivatives evaluated in a multiplicity code are restricted by their (total) degree. Consequently, we let $\F_q[\vec{X}]_d$ denote the vector space of polynomials in $\F_q[\vec{X}]$ that have degree at most~$d$. We index the derivatives of order less than $s$ by the set $S_{s,n}=\{\vec{s}\in\N^n\mid\left|\vec{s}\right|<s\}$, and let $\sigma_{s,n}$ denote its cardinality. Then for $d,s\in\N$ such that $d<sq$, the multiplicity code $\Mult^s_d$ is defined to be the image of the map
\begin{align*}
	\eval^s_d\ :\ \F_q[\Vec{X}]_d
	&\ \rightarrow\ 
	\left(\F^{\sigma_{s,n}}_q\right)^{q^n}\\
	F&\ \mapsto\ 
	\left(
		\left(
			H\left(F,\vec{t}\right)(\vec{\alpha}_{\vec{j}})
		\right)_{\vec{t}\in S_{s,n}}
	\right)_{\vec{j}\in[q]^n}.
\end{align*}
Thus, the multiplicity code $\Mult^s_d$ is a vector space over $\F_q$ of dimension $\binom{n+d}{n}$, while its minimum distance is at least $(1-d/(sq))q^n$~\cite[Lemma~8]{dvir2013} and its information rate is $\binom{n+d}{n}/\left(\binom{n+s-1}{n}q^n\right)$.

\subsection{Systematic encoding of multiplicity codes}\label{sec:encoding}

Given a multiplicity code $\Mult^s_d$, it is natural to consider encoding functions that are $\F_q$-linear functions from~$\F^k_q$ onto the code, where $k=\binom{n+d}{n}$ is the code's dimension. The elements of~$\F^k_q$ are then called the message vectors, or simply messages, of the code. Such an encoding function $\mathrm{enc}:\F^k_q\rightarrow\Mult^s_d$ is systematic if the $i$th entry of each message vector $m$, for $i\in\{1,\dotsc,k\}$, always appears in the encoding $\mathrm{enc}(m)$ at some fixed location. Recording these locations yields a set $\Ical\subseteq[q]^n\times S_{s,n}$ such that the map
\begin{align*}
	\eval_\Ical\ :\ \F_q[\Vec{X}]_d
	&\ \rightarrow\ 
	\F^{\binom{n+d}{n}}_q\\
	F&\ \mapsto\ 
	\left(
		H\left(F,\vec{t}\right)(\vec{\alpha}_{\vec{j}})
	\right)_{(\vec{j},\vec{t})\in\Ical}
\end{align*}
is a bijection. Conversely, a set $\Ical\subseteq[q]^n\times S_{s,n}$ such that the map $\eval_\Ical$ is a bijection induces a systematic encoding function
\begin{equation*}
	\eval^s_d\circ\eval^{-1}_\Ical:\F^{\binom{n+d}{n}}_q\rightarrow\Mult^s_d.
\end{equation*}
Indeed, the function is systematic since the entries of a message vector each reappear in its encoding as the value of some fixed derivative. Such a set $\Ical\subseteq[q]^n\times S_{s,n}$ is called an interpolating set~\cite[Appendix~A]{kopparty2012} or an information set~\cite[Definition~4]{augot2015} of the multiplicity code $\Mult^s_d$.

Kopparty~\cite[Appendix~A]{kopparty2012} provides a method of constructing information sets, and thus a construction of systematic encoding functions, for multiplicity codes. However, Kopparty does not provide explicit examples of the construction. Augot, Levy-dit-Vehel and Ng\^{o}~\cite{augot2015} subsequently provide an explicit family of information sets by supplementing Kopparty's construction with a result of Key, McDonough and Mavron~\cite[Theorem~1]{key2006}.

\begin{theorem}[\cite{key2006,kopparty2012,augot2015}]\label{thm:info_set} For $d,s\in\N$ such that $d<sq$,
\begin{equation*}
	\Ical_{d,n}
	=\left\{
		(\vec{i},\vec{s})\in[q]^n\times\N^n
		\mid
		\left|\vec{i}+\vec{s}q\right|\leq d
	\right\}
\end{equation*}
is an information set of $\Mult^s_d$.
\end{theorem}

We let $\mathrm{enc}^s_d=\eval^s_d\circ\eval^{-1}_{\Ical_{d,n}}$ denote the systematic encoding function of $\Mult^s_d$ provided by Theorem~\ref{thm:info_set}. A codeword of a multiplicity code $\Mult^s_d$ contains $q^n$ elements of $\F^{\sigma_{s,n}}_q$, and thus contains $\sigma_{s,n}q^n$ field elements in total. Consequently, if the encoding function $\mathrm{enc}^s_d$ is to be used in the private information retrieval protocol of Augot, Levy-dit-Vehel and Shikfa~\cite{augot2014}, then it is important that the function may be evaluated in time that is close to linear in $\sigma_{s,n}q^n$. Augot, Levy-dit-Vehel and Ng\^{o}~\cite[Appendix]{augot2015} show that $\mathrm{enc}^s_d$ can be evaluated in $\tilde{\bigO}(\sigma_{s,n}^3q^n+k^2)$ operations in~$\F_q$, where $k=\binom{n+d}{n}$ and the notation $\tilde{\bigO}({}\cdot{})$ indicates that polylogarithmic factors are omitted from the complexity. The quadratic dependency on the dimension of the code means that their algorithm is not suitable for use in the private information retrieval context, where $k\log_2q$ must be greater than or equal to the number of bits in the database, and $q$ is the number of (non-colluding) servers. However, we note that the cost of evaluating $\mathrm{enc}^s_d$ with their algorithm can be reduced to $\tilde{\bigO}(\sigma_{s,n}^3q^n)$ operations in~$\F_q$ by replacing the matrix--vector products they use to perform multivariate interpolation with the quasi-linear time interpolation algorithm of van der Hoeven and Schost~\cite{hoeven2013}.

\subsection{Our contribution} 

In Sections~\ref{sec:low_rate} and~\ref{sec:high_rate}, we present two algorithms that evaluate the encoding function $\mathrm{enc}^s_d$ in $\bigO(\sigma_{s,n}q^nn\log^2(sq)\log\log(sq))$, or more simply $\tilde{\bigO}(\sigma_{s,n}q^n)$, operations in $\F_q$. The algorithm of Section~\ref{sec:low_rate} combines fast polynomial interpolation and evaluation algorithms to first invert the map $\eval_{\Ical_{d,n}}$ then evaluate~$\eval^s_d$. The algorithm of Section~\ref{sec:high_rate} follows a similar interpolation--evaluation approach, but aims to trade a more expensive interpolation step for a cheaper evaluation step. While the two encoding algorithms achieve the same asymptotic complexity, comparing lower order terms of their complexities suggests that they outperform each other at opposing ends of the rate spectrum, with the algorithm of Section~\ref{sec:low_rate} being faster for low-rate codes. Consequently, the two encoding algorithms provide complementary practical performance. 

For the private information retrieval protocol of Augot, Levy-dit-Vehel and Shikfa~\cite{augot2014} one desires to use multiplicity codes with high rates in order to obtain small storage overheads. However, storage overhead must be balanced with other aspects of the protocol when choosing parameters for the codes. The problem of parameter selection is yet to be addressed in the literature, and it remains unclear as to which rates will occur in practice. We do not address this problem here since it is out of the scope of the paper. As a result, we are prevented from determining if one of the two encoding algorithms is better suited to this application.

The interpolation and evaluation algorithms that make up the two encoding algorithms have their origins in the quasi-linear time multivariate interpolation and evaluation algorithms of van der Hoeven and Schost~\cite{hoeven2013}. In Section~\ref{sec:hermite}, we generalise their algorithms to address certain multivariate Hermite interpolation and evaluation problems. Thus, we provide algorithms for recovering multivariate polynomials from their Hasse derivatives, as well as for the inverse problem of computing their derivatives.

The algorithms of van der Hoeven and Schost are recursive in nature, reducing each problem to multiple instances of the same problem in a single variable. Solving the univariate problems in quasi-linear time, then leads to an overall quasi-linear time algorithm. Our Hermite interpolation and evaluation algorithms similarly reduce the multivariate problems to multiple instances of the univariate problems. Applying the quasi-linear time algorithms of Chin~\cite{chin1976} to these univariate instances then yields multivariate algorithms with quasi-linear complexity.

\subsection*{Conventions}

We let $\mult:\N\setminus\{0\}\rightarrow\N$ denote a function such that two univariate polynomials over $\F_q$ of degree less than $k$ can be multiplied in $\mult(k)$ operations in $\F_q$. For example, the algorithm of Cantor and Kaltofen~\cite{cantor1991} implies that $\mult(k)$ may be taken to be in $\bigO(k\log k\log\log k)$.  Throughout the paper, we assume that $\mult(k)/k$ is a nondecreasing function of $k$.

We make frequently use of the shorthand vector notation $(f_{\vec{i}})_{\vec{i}\in I}$ for sets $I\subseteq\N^\ell$. So that this notation is well-defined, we order the entries $f_{\vec{i}}$ by increasing weight~$\left|\vec{i}\right|$ of their multi-indices, with ties broken lexicographically. Similarly, for sets $\mathcal{I}\subseteq[q]^\ell\times\N^\ell$, we assume that the entries of a vector $(f_{(\vec{i},\vec{s})})_{(\vec{i},\vec{s})\in\mathcal{I}}$ are ordered by increasing $\left|\vec{i}+\vec{s}q\right|$, with ties broken by comparing the vectors $\vec{i}+\vec{s}q$ lexicographically.

For $\vec{i}=(i_1,\dotsc,i_n)\in\Z^n$ and $j\in\N\setminus\{0\}$, we define $\vec{i}\bdiv{j}=(\floor{i_1/j},\dotsc,\floor{i_n/j})$ and $\vec{i}\bmod{j}=\vec{i}-(\vec{i}\bdiv{j})j$. Similarly, for $F,G\in\F_q[X]$ such that $\deg G>0$, we write $F\bmod{G}$ for the residue of $F$ modulo $G$ that has degree less than $\deg G$.

\section{Multivariate Hermite interpolation and evaluation}\label{sec:hermite}

The interpolation algorithm of van der Hoeven and Schost~\cite{hoeven2013}, when applied over~$\F_q$, takes as an input a vector of field elements $(m_{\vec{j}})_{\vec{j}\in I}$ for some $I\subseteq[q]^n$, and returns the unique polynomial $F\in\F_q[\vec{X}]$ that has support contained in $I$ and satisfies $F(\vec{\alpha}_{\vec{j}})=m_{\vec{j}}$ for $\vec{j}\in I$. Their evaluation algorithm performs the inverse computation, evaluating a polynomial with support contained in $I$ at the points $\vec{\alpha}_{\vec{j}}$ for $\vec{j}\in I$. Both algorithms require $I$ to be an initial segment for the partial order~$\leq$ on $\N^n$ defined by $\vec{i}\leq\vec{j}$ if and only if $\vec{j}-\vec{i}\in\N^n$: a subset $I\subseteq\N^n$ is then an initial segment if it is nonempty and contains all $\vec{i}\in\N^n$ such that $\vec{i}\leq\vec{j}$ for some $\vec{j}\in I$. 

For $I\subseteq\N^n$, let $\F_q[\vec{X}]_I$ denote the vector space of polynomials in $\F_q[\vec{X}]$ that have support contained in $I$. Then a key feature of the algorithms of van der Hoeven and Schost is the representation of polynomials in $\F_q[\vec{X}]_I$, where $I\subseteq[q]^n$ is an initial segment, with respect to a multivariate Newton basis. This basis consists of the polynomials
\begin{equation*}
	N_{\vec{i}}(\vec{X})=N_{i_1}(X_1)\dotsm N_{i_n}(X_n)
	\quad\text{for $\vec{i}=(i_1,\dotsc,i_n)\in I$},
\end{equation*}
where
\begin{equation*}
	N_i(X)=(X-\alpha_0)\dotsm(X-\alpha_{i-1})
	\quad\text{for $i\in[q]$}
\end{equation*}
are the Newton polynomials associated with the enumeration of the field. The Newton basis polynomial $N_{\vec{i}}$ vanishes at all points $\vec{\alpha}_{\vec{j}}$ with $\vec{j}\ngeq\vec{i}$, allowing van der Hoeven and Schost to address the interpolation and evaluation problems one variable at a time in a manner similar to the earlier work of Pan~\cite{pan1994}. In doing so, they obtain algorithms for both problems that each perform $\bigO(\left|I\right|n\log^2\left|I\right|\log\log\left|I\right|)$ field operations.

In this section, we generalise the interpolation and evaluation algorithms of van der Hoeven and Schost to address multivariate Hermite interpolation and evaluation problems. The generalised algorithms yield analogous complexities to those of the algorithms of van der Hoeven and Schost. Thus, they allow the fast recovery of polynomials from the values of their Hasse derivatives, in addition to allowing the fast evaluation of their derivatives.

\subsection{Hermite interpolation and evaluation}\label{sec:hermite_definition}

We generalise the interpolation and evaluation problems considered by van der Hoeven and Schost through generalising the use of the multivariate Newton basis. To allow initial segments that are not contained in $[q]^n$, we extend the definition of the Newton basis by introducing repeated roots to the basis polynomials. We define
\begin{equation*}
	N_i(X)=\prod^{i-1}_{j=0}\left(X-\alpha_{j\bmod{q}}\right)
	\quad\text{for $i\in\N$},
\end{equation*}
and $N_{\vec{i}}(\vec{X})=N_{i_1}(X_1)\dotsm N_{i_n}(X_n)$ for $\vec{i}=(i_1,\dotsc,i_n)\in\N^n$. Then, for $\vec{i}\in\N^n$, the polynomial $N_{\vec{i}}$ may be written in the form $\sum_{\vec{k}\leq\vec{i}}n_{\vec{k}}\vec{X}^{\vec{k}}$ with coefficients $n_{\vec{k}}\in\F_q$ and $n_{\vec{i}}=1$. Therefore, under the extended definition we retain the property that $\{N_{\vec{i}}\mid\vec{i}\in I\}$ is a basis of $\F_q[\vec{X}]_I$ when $I\subseteq\N^n$ is an initial segment. However, having introduced repeated roots to the basis polynomials, the vanishing property of the Newton basis now extends to include the Hasse derivatives of the basis polynomials.

\begin{lemma}\label{lem:basis_property} For $\vec{i}\in\N^n$ and $(\vec{j},\vec{t})\in[q]^n\times\N^n$, we have $H(N_{\vec{i}},\vec{t})(\vec{\alpha}_{\vec{j}})=0$ if $\vec{j}+\vec{t}q\ngeq\vec{i}$, and $H(N_{\vec{i}},\vec{t})(\vec{\alpha}_{\vec{j}})\neq 0$ if $\vec{j}+\vec{t}q=\vec{i}$.
\end{lemma}
\begin{proof} It is sufficient to prove the lemma for all $N_{\vec{i}+\vec{s}q}$ with $(\vec{i},\vec{s})\in[q]^n\times\N^n$. Let $(\vec{i},\vec{s}),(\vec{j},\vec{t})\in[q]^n\times\N^n$ with $\vec{i}=(i_1,\dotsc,i_n)$, $\vec{s}=(s_1,\dotsc,s_n)$, $\vec{j}=(j_1,\dotsc,j_n)$ and $\vec{t}=(t_1,\dotsc,t_n)$. Then, for algebraically independent indeterminates $Z_1,\dotsc, Z_n$ over $\F_q[\vec{X}]$, the definition of the Hasse derivative implies that $H(N_{\vec{i}+\vec{s}q},\vec{t})(\vec{\alpha}_{\vec{j}})$ is equal to the coefficient of $\vec{Z}^{\vec{t}}=Z^{t_1}_1\dotsm Z^{t_n}_n$ in the polynomial $N_{\vec{i}+\vec{s}q}(\vec{Z}+\vec{\alpha}_{\vec{j}})\in\F_q[\vec{Z}]=\F_q[Z_1,\dotsc, Z_n]$. 

For $\ell\in\{1,\dotsc,n\}$, let $\varepsilon_\ell:\N\rightarrow\{0,1\}$ be the indicator function defined by $\varepsilon_\ell(k)=1$ if and only if $k<i_\ell$. Then
\begin{equation*}
	N_{i_\ell+s_\ell q}(X_\ell)
	=\prod_{k\in[q]}
	\left(X_\ell-\alpha_k\right)^{s_\ell+\varepsilon_\ell(k)}
	\quad\text{for $\ell=1,\dotsc,n$}.
\end{equation*}
Letting $\vec{\varepsilon}=(\varepsilon_1(j_1),\dotsc,\varepsilon_n(j_n))$, it follows that
\begin{equation}\label{eqn:basis_property}
	N_{\vec{i}+\vec{s}q}\left(\vec{Z}+\vec{\alpha}_{\vec{j}}\right)
	=\vec{Z}^{\vec{s}+\vec{\varepsilon}}
	\prod^n_{\ell=1}
	\prod_{k\in[q]\setminus\{j_\ell\}}
	\left(Z_\ell+\alpha_{j_\ell}-\alpha_k\right)^{s_\ell+\varepsilon_\ell(k)}.
\end{equation}
If $\vec{t}\geq\vec{s}+\vec{\varepsilon}$, then $\vec{j}+\vec{t}q\geq(\vec{j}+\vec{\varepsilon}q)+\vec{s}q\geq\vec{i}+\vec{s}q$. Therefore, if $\vec{j}+\vec{t}q\ngeq\vec{i}+\vec{s}q$, then $H(N_{\vec{i}+\vec{s}q},\vec{t})(\vec{\alpha}_{\vec{j}})=0$ since $\vec{t}\ngeq\vec{s}+\vec{\varepsilon}$ and $\vec{Z}^{\vec{s}+\vec{\varepsilon}}$ divides $N_{\vec{i}+\vec{s}q}(\vec{Z}+\vec{\alpha}_{\vec{j}})$ in $\F_q[\vec{Z}]$. If $\vec{j}+\vec{t}q=\vec{i}+\vec{s}q$, then $(\vec{j},\vec{t})=(\vec{i},\vec{s})$ and $\vec{\varepsilon}=\vec{0}$. By substituting into~\eqref{eqn:basis_property} and computing the coefficient of $\vec{Z}^{\vec{s}}=\vec{Z}^{\vec{s}+\vec{\varepsilon}}$, we find that
\begin{equation*}
	H\left(N_{\vec{i}+\vec{s}q},\vec{s}\right)(\vec{\alpha}_{\vec{i}})
	=\prod^n_{\ell=1}
	\prod_{k\in[q]\setminus\{i_\ell\}}
	\left(\alpha_{i_\ell}-\alpha_k\right)^{s_\ell+\varepsilon_\ell(k)},
\end{equation*}
which is nonzero.
\end{proof}

In the interpolation and evaluation problems considered by van der Hoeven and Schost, the initial segment $I$ is the support of both the polynomials and the evaluation points. In order to maintain this property when generalising these problems, we define $\evH(F,\vec{i})=H(F,\vec{i}\bdiv{q})(\vec{\alpha}_{\vec{i}\bmod{q}})$ for $\vec{i}\in\N^n$ and $F\in\F_q[\vec{X}]$. Then our problem of Hermite interpolation takes a vector $(m_{\vec{j}})_{\vec{j}\in I}$ of field elements for some finite initial segment $I\subseteq\N^n$ and asks that we compute the polynomial $F\in\F_q[\vec{X}]_I$ that satisfies $\evH(F,\vec{j})=m_{\vec{j}}$ for $\vec{j}\in I$. Our Hermite evaluation problem is the inverse problem, asking for the computation of the vector $(\evH(F,\vec{j}))_{\vec{j}\in I}$ when given a polynomial $F\in\F_q[\vec{X}]_I$. Importantly, Lemma~\ref{lem:basis_property} implies that $\evH(N_{\vec{i}},\vec{j})=0$ for all $\vec{i},\vec{j}\in\N^n$ such that $\vec{i}\nleq\vec{j}$, allowing us to address both problems by generalising the algorithms of van der Hoeven and Schost. Existence and uniqueness of a solution to the Hermite interpolation problem is provided by the following lemma.

\begin{lemma}\label{lem:existence_uniqueness} Let $I\subseteq\N^n$ be a finite initial segment and $(m_{\vec{j}})_{\vec{j}\in I}\in\F^{\left|I\right|}_q$. Then there exists a unique polynomial $F\in\F_q[\vec{X}]_I$ such that $\evH(F,\vec{j})=m_{\vec{j}}$ for $\vec{j}\in I$.
\end{lemma}
\begin{proof} If $I\subseteq\N^n$ is a finite initial segment, then $\F_q[\vec{X}]_I$ and $\F^{\left|I\right|}_q$ are $\left|I\right|$-dimensional $\F_q$-vector spaces, and $I\subseteq S_{s,n}$ for $s>\max_{\vec{i}\in I}\left|\vec{i}\right|$. Therefore, it is sufficient to prove following statement: for all positive $s\in\N$, if $I\subseteq S_{s,n}$ is an initial segment, then the homomorphism $\eval_I:\F_q[\vec{X}]_I\rightarrow\F^{\left|I\right|}_q$ given by $F\mapsto(\evH(F,\vec{j}))_{\vec{j}\in I}$ is injective. We prove this statement by induction on $s$. The statement holds trivially for $s=1$, since $\{\vec{0}\}$ is the only initial segment contained in $S_{1,n}$, and $\eval_{\{\vec{0}\}}:\F_q\rightarrow\F_q$ is the identity map. Therefore, suppose that the statement is true for some integer $s\geq 1$. Let $I\subseteq S_{s+1,n}$ be an initial segment, and $F\in\F_q[\vec{X}]_I$ such that $\evH(F,\vec{j})=0$ for $\vec{j}\in I$. Then to complete the proof of the lemma, it is sufficient to show that $F$ is equal to zero.
	
Let $J=I\cap S_{s,n}$. Then $J$ is an initial segment since $I$ and $S_{s,n}$ are initial segments. Moreover, if $\vec{i}\in I\setminus J$, then its weight $\left|\vec{i}\right|$ is maximal amongst the elements of $I$. Consequently, if $\vec{i}\in I\setminus J$, then $\vec{j}\ngeq\vec{i}$ for $\vec{j}\in I\setminus\{\vec{i}\}$. As Hasse derivatives are linear functions and evaluation is a homomorphism, the functions $\evH({}\cdot{},\vec{j}):\F_q[\vec{X}]\rightarrow\F_q$ for $\vec{j}\in\N$ are linear. Therefore, if we write $F=\sum_{\vec{i}\in I}f_{\vec{i}}N_{\vec{i}}$ such that $f_{\vec{i}}\in\F_q$ for all $\vec{i}\in I$, then Lemma~\ref{lem:basis_property} implies that
\begin{equation*}
	0
	=\evH\left(\sum_{\vec{i}\in J}f_{\vec{i}}N_{\vec{i}},\vec{j}\right)
	+\sum_{\vec{i}\in I\setminus J}f_{\vec{i}}\evH\left(N_{\vec{i}},\vec{j}\right)
	=\evH\left(\sum_{\vec{i}\in J}f_{\vec{i}}N_{\vec{i}},\vec{j}\right)
	\quad\text{for $\vec{j}\in J$}.
\end{equation*}
As $J\subseteq S_{s,n}$ is an initial segment, the induction hypothesis implies that $\sum_{\vec{i}\in J}f_{\vec{i}}N_{\vec{i}}$ is equal to zero. Applying Lemma~\ref{lem:basis_property} once again, it follows that
\begin{equation*}
	0
	=\sum_{\vec{i}\in I\setminus J}f_{\vec{i}}\evH\left(N_{\vec{i}},\vec{j}\right)
	=f_{\vec{j}}\evH\left(N_{\vec{j}},\vec{j}\right)
	\quad\text{for $\vec{j}\in I\setminus J$}.
\end{equation*}
Moreover, the lemma states that $\evH(N_{\vec{j}},\vec{j})\neq 0$ for $\vec{j}\in I\setminus J$. Therefore, $f_{\vec{i}}=0$ for $\vec{i}\in I\setminus J$. Hence, $F$ is equal to zero.
\end{proof}

Define $\kappa_n:[q]^n\times\N^n\rightarrow\N^n$ by $(\vec{i},\vec{s})\mapsto\vec{i}+\vec{s}q$. Then $\kappa_n(\Ical_{d,n})=\{\vec{i}\in\N^n\mid\left|\vec{i}\right|\leq d\}$ is a finite initial segment for $d\in\N$. Moreover, for $F\in\F_q[\vec{X}]_d=\F_q[\vec{X}]_{\kappa_n(\Ical_{d,n})}$ we have $(H(F,\vec{t})(\vec{\alpha}_{\vec{j}}))_{(\vec{j},\vec{t})\in\Ical_{d,n}}=(\evH(F,\vec{j}))_{\vec{j}\in\kappa_n(\Ical_{d,n})}$. Thus, the problem of computing the polynomial that corresponds to a message vector of a multiplicity code $\Mult^s_d$ is an instance of the Hermite interpolation problem with initial segment $I=\kappa_n(\Ical_{d,n})$. Similarly, the problem of encoding a polynomial as a codeword in $\Mult^s_d$, i.e., evaluating the map $\eval^s_d$ for some polynomial of degree at most $d<sq$, is an instance of the Hermite evaluation problem with initial segment $I=\kappa_n([q]^n\times S_{s,n})\supset\kappa_n(\Ical_{d,n})$. In Section~\ref{sec:low_rate}, we apply the fast algorithms developed in this section to these two instances to obtain a fast systematic encoding algorithm for low-rate codes, while in Section~\ref{sec:high_rate}, the interpolation algorithm is applied with $I=\kappa_n([q]^n\times S_{s,n})$ as part of the encoding algorithm for higher rate codes.

As we have no need to represent polynomials with respect to the monomial basis during encoding, we only require that the output of the interpolation algorithm and input of the evaluation algorithm are written on the Newton basis. Consequently, if $I\subseteq\N^n$ is an initial segment, then we write $F\dashv(N_{\vec{i}})_{\vec{i}\in I}$ for the vector of coefficients of $F\in\F_q[\vec{X}]_I$ when written on the basis $(N_{\vec{i}})_{\vec{i}\in I}$. That is, if $F=\sum_{\vec{i}\in I}f_{\vec{i}}N_{\vec{i}}$ such that the coefficients $f_{\vec{i}}\in\F_q$, then $F\dashv(N_{\vec{i}})_{\vec{i}\in I}=(f_{\vec{i}})_{\vec{i}\in I}$. Similarly, we write $F\dashv(\vec{X}^{\vec{i}})_{\vec{i}\in I}$ for the coefficient vector of $F$ when written on the monomial basis. To allow us to bound the size of a finite initial segment in each of its dimensions, we extend the notation $[q]$ by defining $[s]=\{0,1,\dotsc,s-1\}$ for positive $s\in\N$. Using this notation, we can state the main result of this section as follows.

\begin{theorem}\label{thm:multivariate_hermite} Let $I\subseteq\N^n$ be an initial segment such that $I\subseteq[s_1]\times\dotsb\times[s_n]$ for positive integers $s_1,\dotsc,s_n$. Then given the vector $(\evH(F,\vec{j}))_{\vec{j}\in I}$ for some polynomial $F\in\F_q[\vec{X}]_I$, the vector $F\dashv(N_{\vec{i}})_{\vec{i}\in I}$ can be computed in
\begin{equation}\label{eqn:multivariate_hermite}
	\bigO\left(
		\left(
			\frac{\mult(s_1)\log s_1}{s_1}
			+\dotsb+
			\frac{\mult(s_n)\log s_n}{s_n}
		\right)
		\left|I\right|
	\right)
\end{equation}
operations in $\F_q$. Conversely, given the vector $F\dashv(N_{\vec{i}})_{\vec{i}\in I}$ for some polynomial $F\in\F_q[\vec{X}]_I$, the vector $(\evH(F,\vec{j}))_{\vec{j}\in I}$ can be computed within the same bound on the number of operations in $\F_q$.
\end{theorem}

Theorem~\ref{thm:multivariate_hermite} directly generalises the bounds obtained by van der Hoeven and Schost~\cite[Propostions~2 and~3]{hoeven2013} for interpolation and evaluation. By letting $s_1,\dotsc,s_n$ equal $\left|I\right|$, and taking $\mult(k)$ to be in $\bigO(k\log k\log\log k)$, the bound~\eqref{eqn:multivariate_hermite} simplifies to $\bigO(\left|I\right|n\log^2\left|I\right|\log\log\left|I\right|)$, matching the bound for the algorithms of van der Hoeven and Schost stated at the beginning of the section.

In other settings it may be preferable to have the output of the Hermite interpolation algorithm or the input of the Hermite evaluation algorithm represented with respect to the monomial basis. For univariate polynomials, conversion between the Newton and monomial bases can be performed in quasi-linear time by the algorithms discussed in the next section. These algorithms extend to multivariate polynomials by applying the approach of van der Hoeven and Schost~\cite[Section~4]{hoeven2013}. Using these algorithms, it is possible to preserve the bound~\eqref{eqn:multivariate_hermite} while having the input and output polynomials of the Hermite interpolation and evaluation algorithms given on the monomial basis.

The remainder of this section is devoted to proving Theorem~\ref{thm:multivariate_hermite}. We begin in the next section by reviewing existing fast algorithms for solving the Hermite interpolation and evaluation problems in univariate case. Then we complete the proof of the theorem by generalising the multivariate interpolation and evaluation algorithms of van der Hoeven and Schost in Section~\ref{sec:multivariate}.

\subsection{Univariate algorithms}\label{sec:univariate}

Hermite interpolation and evaluation for univariate polynomials can be performed in quasi-linear time with respect to the monomial basis by the algorithms of Chin~\cite{chin1976}.  In these algorithms, derivative is taken to mean the formal derivative rather than the Hasse derivative, as required here. However, by using the fact that the $i$th formal derivative is equal to $i!$ times the $i$th Hasse derivative, it is readily shown that only superficial changes to Chin's algorithms are required to allow them to work with the Hasse derivative. We note that the convolution-based algorithm of Aho, Steiglitz and Ullman~\cite{aho1975} that is used by Chin to compute Taylor shifts of polynomials cannot be used if the characteristic of the field is not greater than their degrees. In this case, the convolution-based algorithm may be replaced by the algorithm of Olshevsky and Shokrollahi~\cite[Section~4.2]{olshevsky2000} (see also~\cite{gathen1990,gathen1997}), which is slower by a logarithmic factor.

Each finite initial segment in $\N$ is of the form $[s]$ for some positive integer~$s$. For the Hermite interpolation and evaluation problems defined by these initial segments, applying Chin's algorithms with modifications just described provides the following complexity bounds.

\begin{lemma}\label{lem:univariate_hermite} Let $s\in\N$ be positive and $F\in\F_q[X]_{[s]}$. Then given $(\evH(F,j))_{j\in[s]}$, the vector $F\dashv(X^i)_{i\in[s]}$ can be computed in $\bigO(\mult(s)\log s)$ operations in $\F_q$. Conversely, given $F\dashv(X^i)_{i\in[s]}$, the vector $(\evH(F,j))_{j\in[s]}$ can be computed in $\bigO(\mult(s)\log s)$ operations in $\F_q$.
\end{lemma}

Closely related alternatives to Chin's algorithms that provide the same complexity bounds are given by Olshevsky and Shokrollahi~\cite{olshevsky2000} and texts~\cite[Chapter~3]{pan2001}, \cite[Chapter~1, Section~4]{bini1994} and~\cite[Exercise~3.14]{burgisser1997}. In situations where precomputation is permitted, the asymptotic complexity of these algorithms and Chin's algorithms may be improved upon by using the techniques described by van der Hoeven~\cite{hoeven2016}.

Combining Lemma~\ref{lem:univariate_hermite} with the following result of Gerhard~\cite{gerhard2000} completes the proof of Theorem~\ref{thm:multivariate_hermite} for the univariate case.

\begin{lemma}\label{lem:univariate_conversion} Let $s\in\N$ be positive and $F\in\F_q[X]_{[s]}$. Then given $F\dashv(N_i)_{i\in[s]}$, the vector $F\dashv(X^i)_{i\in[s]}$ can be computed in $\bigO(\mult(s)\log s)$ operations in $\F_q$. Conversely, given $F\dashv(X^i)_{i\in[s]}$, the vector $F\dashv(N_i)_{i\in[s]}$ can be computed in $\bigO(\mult(s)\log s)$ operations in $\F_q$.
\end{lemma}

When converting from the monomial basis to the Newton basis, the algorithm of Gerhard is improved upon in practice by the algorithm of Bostan and Schost~\cite{bostan2005}.

\subsection{Multivariate algorithms}\label{sec:multivariate}

By design, the Hermite interpolation and evaluation problems allow the algorithms of van der Hoeven and Schost to be generalised in a straightforward manner. However, we follow a slightly different course by presenting the generalised algorithms in an iterative, rather than recursive, form. This small change is used to simplify the description of modifications to the algorithms that are made in the encoding context.

We begin by introducing some geometric operations on initial segments. For $I\subseteq\N^n$ and $\vec{i}=(i_1,\dotsc,i_\ell)\in\N^\ell$ such that $1\leq\ell<n$, define
\begin{equation*}
	\lambda(I,\vec{i})
	=\left\{
		(j_1,\dotsc,j_{n-\ell})\in\N^{n-\ell}
		\mid (j_1,\dotsc,j_{n-\ell},i_1,\dotsc,i_\ell)\in I
	\right\}
\end{equation*}
and
\begin{equation*}
	\rho(I,\vec{i})
	=\left\{
		(j_1,\dotsc,j_{n-\ell})\in\N^{n-\ell}
		\mid (i_1,\dotsc,i_\ell,j_1,\dotsc,j_{n-\ell})\in I
	\right\}.
\end{equation*}
Let $\vec{0}_\ell$ denote the $\ell$-dimensional vector of zeros. Then, given an initial segment $I\subseteq\N^n$ and a positive integer $\ell<n$, the set $\lambda(I,\vec{0}_\ell)$ is the projection of $I$ onto the $(i_1,\dotsc,i_{n-\ell})$-coordinate plane, while $\rho(I,\vec{0}_\ell)$ is the projection of $I$ onto the $(i_{\ell+1},\dotsc,i_n)$-coordinate plane. Consequently, if $F\in\F_q[\vec{X}]_I$ has coefficient vector $F\dashv(N_{\vec{i}})_{\vec{i}\in I}=(f_{\vec{i}})_{\vec{i}\in I}$, then
\begin{equation}\label{eqn:collected_coefficients}
	F=\sum_{i_n\in\rho(I,\vec{0}_{n-1})}F_{i_n}(X_1,\dotsc,X_{n-1})N_{i_n}(X_n)
\end{equation}
where
\begin{equation}\label{eqn:coefficients}
	F_{i_n}\left(X_1,\dotsc,X_{n-1}\right)
	=\sum_{(i_1,\dotsc,i_{n-1})\in\lambda\left(I,i_n\right)}
	f_{(i_1,\dotsc,i_{n-1},i_n)}
	N_{i_1}(X_1)\dotsm N_{i_{n-1}}(X_{n-1})
\end{equation}
for $i_n\in\rho(I,\vec{0}_{n-1})$. For $F\in\F_q[\vec{X}]$ and $i_n\in\N$, we define $F_{i_n}\in\F_q[X_1,\dotsc,X_{n-1}]$ to be the polynomial given by \eqref{eqn:coefficients} for $I=\N^n$ and $(f_{\vec{i}})_{\vec{i}\in\N^n}=F\dashv(N_{\vec{i}})_{\vec{i}\in\N^n}$. Then \eqref{eqn:collected_coefficients} and \eqref{eqn:coefficients} still hold whenever $F\in\F_q[\vec{X}]_I$ for some initial segment $I\subseteq\N^n$, but the definition of $F_{i_n}$ is now independent of $I$. 

We base our Hermite interpolation and evaluation algorithms on the following analogue of \cite[Proposition~1]{hoeven2013} for the functions $\evH({}\cdot{},\vec{j})$.

\begin{lemma}\label{lem:recursion} Let $I\subseteq\N^n$ be a finite initial segment and $F\in\F_q[\vec{X}]_I$. Then
\begin{equation*}
	\evH(F,\vec{j})
	=\evH\left(
		\sum_{i_n\in\rho\left(I,(j_1,\dotsc,j_{n-1})\right)}
		\evH\left(F_{i_n},(j_1,\dotsc,j_{n-1})\right)
		N_{i_n}(X_n)
	,j_n\right)
\end{equation*}
for all $\vec{j}=(j_1,\dotsc,j_n)\in I$.
\end{lemma}
\begin{proof}We begin the proof by establishing a multiplicative property of the functions $\evH({}\cdot{},\vec{j})$. Let $U\in\F_q[X_1,\dotsc,X_{n-1}]$ and $V\in\F_q[X_n]$. Then it follows from the definition of the Hasse derivative that $H(UV,\vec{s})=H(U,(s_1,\dotsc,s_{n-1}))H(V,s_n)$ for $\vec{s}=(s_1,\dotsc,s_n)\in\N^n$. As evaluation is a homomorphism, we conclude that $\evH(UV,\vec{j})=\evH(U,(j_1,\dotsc,j_{n-1}))\evH(V,j_n)$ for $\vec{j}=(j_1,\dotsc,j_n)\in\N^n$.
	
Suppose now that $I\subseteq\N^n$ is a finite initial segment, $F\in\F_q[\vec{X}]_I$ and $\vec{j}=(j_1,\dotsc,j_n)\in I$. Then \eqref{eqn:collected_coefficients} holds, from which it follows that
\begin{equation*}
	\evH\left(F,\vec{j}\right)
	=\sum_{i_n\in\rho(I,\vec{0}_{n-1})}
	\evH\left(F_{i_n},(j_1,\dotsc,j_{n-1})\right)
	\evH\left(N_{i_n},j_n\right).
\end{equation*}
As $\evH\left({}\cdot{},j_n\right):\F_q[X_n]\rightarrow\F_q$ is a linear function, the proof of the lemma will be complete if we show that $\evH(N_{i_n},j_n)=0$ for $i_n\in\rho(I,\vec{0}_{n-1})\setminus\rho\left(I,(j_1,\dotsc,j_{n-1})\right)$. If $i_n\in\rho(I,\vec{0}_{n-1})$ and $i_n\leq j_n$, then $i_n\in\rho\left(I,(j_1,\dotsc,j_{n-1})\right)$ since $I$ is an initial segment and $(j_1,\dotsc,j_{n-1},i_n)\leq(j_1,\dotsc,j_{n-1},j_n)\in I$. As a result, $i_n>j_n$ for $i_n\in\rho(I,\vec{0}_{n-1})\setminus\rho\left(I,(j_1,\dotsc,j_{n-1})\right)$. Hence, Lemma~\ref{lem:basis_property} implies that $\evH(N_{i_n},j_n)=0$ for $i_n\in\rho(I,\vec{0}_{n-1})\setminus\rho\left(I,(j_1,\dotsc,j_{n-1})\right)$.
\end{proof}

Lemma~\ref{lem:recursion} sets up a natural recursive approach to the Hermite interpolation and evaluation problems by reducing each problem to a combination of univariate problems in the variable $X_n$, and the recovery or evaluation of the $(n-1)$-variate polynomials $F_{i_n}$. To allow us to instead present iterative algorithms, we must introduce some additional geometric operations on initial segments. For $n\geq 2$, $I\subseteq\N^n$ and $\ell\in\{1,\dotsc,n\}$, we define
\begin{equation*}
	\pi_\ell(I)
	=\left\{
		(i_1,\dotsc,i_{\ell-1},i_{\ell+1},\dotsc,i_n)
		\mid
		(i_1,\dotsc,i_n)\in I
	\right\}
\end{equation*}
to be the projection of $I$ onto the $(i_1,\dotsc,i_{\ell-1},i_{\ell+1},\dotsc,i_n)$-coordinate plane. For $\vec{i}=(i_1,\dotsc,i_{\ell-1},i_{\ell+1},\dotsc,i_n)\in\pi_\ell(I)$, we define
\begin{equation*}
	\mu_\ell(I,\vec{i})
	=\left\{
		i_\ell\in\N\mid (i_1,\dotsc,i_{\ell-1},i_\ell,i_{\ell+1},\dotsc,i_n)\in I
	\right\}.
\end{equation*}
We extend these definitions to $I\subseteq\N$ by defining $\pi_1(I)=\{0\}$ and $\mu_1(I,0)=I$. When $I\subseteq\N^n$ is an initial segment, so too are the sets $\mu_\ell(I,\vec{i})$ for $\vec{i}\in\pi_\ell(I)$.

The multivariate Hermite evaluation and interpolation algorithms are presented in Algorithms~\ref{alg:multi_hermite_eval} and~\ref{alg:multi_hermite_interp}, respectively. We require that the univariate algorithms they use to be in-place algorithms in the sense that inputs are overwritten by their corresponding output. In particular, the input and output specifications of the univariate algorithms should match those of their corresponding multivariate algorithm for $n=1$. However, we do not impose restrictions on the memory usage of the algorithms, as is usual when defining the notion of ``in-place'', so that any univariate algorithm can be modified to fit this description. We are deliberately non-committal about the choice of univariate algorithms, since any algorithms that solve the univariate problems may be used. One may, of course, take these algorithms to be the corresponding algorithm of Chin with the modifications described in Section~\ref{sec:univariate}, including basis conversion to ensure that input and output polynomials are written on the Newton basis. In particular, it is this combination of algorithms that is used to prove Theorem~\ref{thm:multivariate_hermite}.

\begin{algorithm}[!ht]
	\caption{Multivariate Hermite evaluation}\label{alg:multi_hermite_eval}
	\begin{algorithmic}[1]
		\Require A finite initial segment $I\subseteq\N^n$; and the vector $F\dashv(N_{\vec{i}})_{\vec{i}\in I}=(f_{\vec{i}})_{\vec{i}\in I}$ for some polynomial $F\in\F_q[\vec{X}]_I$.
		
		\Ensure The vector $(f_{\vec{i}})_{\vec{i}\in I}$ equal to $(\evH(F,\vec{j}))_{\vec{j}\in I}$.
		
		\For{$\ell=1,\dotsc,n$}
		
			\For{$\vec{k}=(j_1,\dotsc,j_{\ell-1},i_{\ell+1},\dotsc,i_n)\in\pi_\ell(I)$}
			
				\Statew{3}{Call an in-place univariate Hermite evaluation algorithm on the vector $\left(f_{(j_1,\dotsc,j_{\ell-1},i_\ell,i_{\ell+1},\dotsc,i_n)}\right)_{i_\ell\in\mu_\ell(I,\vec{k})}$.}\label{step:eval}
			
			\EndFor
		
		\EndFor
	\end{algorithmic}
\end{algorithm}

\begin{algorithm}[!ht]
	\caption{Multivariate Hermite interpolation}\label{alg:multi_hermite_interp}
	\begin{algorithmic}[1]
		\Require A finite initial segment $I\subseteq\N^n$; and the vector $(\evH(F,\vec{j}))_{\vec{j}\in I}=(f_{\vec{j}})_{\vec{j}\in I}$ for some polynomial $F\in\F_q[\vec{X}]_I$.
		
		\Ensure The vector $(f_{\vec{j}})_{\vec{j}\in I}$ equal to $F\dashv(N_{\vec{i}})_{\vec{i}\in I}$.
		
		\For{$\ell=n,n-1,\dotsc,1$}
		
			\For{$\vec{k}=(j_1,\dotsc,j_{\ell-1},i_{\ell+1},\dotsc,i_n)\in\pi_\ell(I)$}
			
				\Statew{3}{Call an in-place univariate Hermite interpolation algorithm on the vector $\left(f_{(j_1,\dotsc,j_{\ell-1},j_\ell,i_{\ell+1},\dotsc,i_n)}\right)_{j_\ell\in\mu_\ell(I,\vec{k})}$.}
			
			\EndFor
		
		\EndFor
	\end{algorithmic}
\end{algorithm}

We prove that Algorithm~\ref{alg:multi_hermite_eval} is correct in Lemma~\ref{lem:correctness}. Combining the lemma with Lemma~\ref{lem:existence_uniqueness} then establishes the correctness of Algorithm~\ref{alg:multi_hermite_interp}, since the algorithm simply reverses the steps of the Algorithm~\ref{alg:multi_hermite_eval}, inverting each evaluation along the way.

\begin{samepage}
\begin{lemma}\label{lem:correctness} Algorithm~\ref{alg:multi_hermite_eval} is correct.
\end{lemma}
\begin{proof} We prove the lemma by induction on $n$. If $n=1$, then Algorithm~\ref{alg:multi_hermite_eval} simply calls the univariate algorithm on the input. Accordingly, correctness holds trivially for univariate inputs. It is illustrative to consider the case $n=2$ separately before proceeding by induction. Therefore, suppose that Algorithm~\ref{alg:multi_hermite_eval} is called on a finite initial segment $I\subseteq\N^2$ and the vector $F\dashv(N_{(i_1,i_2)})_{(i_1,i_2)\in I}=(f_{(i_1,i_2)})_{(i_1,i_2)\in I}$ for some $F\in\F_q[X_1,X_2]_I$. Then the first iteration of the outer loop of the algorithm calls the univariate algorithm on each of the vectors
\begin{equation*}
	\left(f_{(i_1,i_2)}\right)_{i_1\in\mu_1(I,i_2)}
	=\left(f_{(i_1,i_2)}\right)_{i_1\in\lambda(I,i_2)}
	\quad\text{for $i_2\in\pi_1(I)=\rho(I,0)$}.
\end{equation*}
Here, $(f_{(i_1,i_2)})_{i_1\in\lambda(I,i_2)}$ is equal to the coefficient vector $F_{i_2}\dashv(N_{i_1})_{i_1\in\lambda(I,i_2)}$. Therefore, after the first iteration of the outer loop has been performed, the input vector $(f_{(i_1,i_2)})_{(i_1,i_2)\in I}$ is equal to $(\evH(F_{i_2},j_1))_{(j_1,i_2)\in I}$. It follows that the second iteration of the outer loop calls the univariate algorithm on each of the vectors
\begin{equation*}
	\left(f_{(j_1,i_2)}\right)_{i_2\in\mu_2(I,j_1)}
	=\left(\evH(F_{i_2},j_1)\right)_{i_2\in\rho(I,j_1)}
	\quad\text{for $j_1\in\pi_2(I)=\lambda(I,0)$}.
\end{equation*}
Thus, Lemma~\ref{lem:recursion} implies that after the second iteration of the outer loop has been performed, we have $(f_{(j_1,i_2)})_{i_2\in\rho(I,j_1)}=(\evH(F,(j_1,j_2)))_{j_2\in\rho(I,j_1)}$ for $j_1\in\lambda(I,0)$. As this is the last iteration of the loop, it follows that the input vector is equal to $(E(F,(j_1,j_2)))_{(j_1,j_2)\in I}$ at the end of the algorithm. Hence, Algorithm~\ref{alg:multi_hermite_eval} is correct for the inputs $I$ and $F\dashv(N_{(i_1,i_2)})_{(i_1,i_2)\in I}$, and the lemma holds for $n=2$.

Suppose now that $n\geq 3$ and Algorithm~\ref{alg:multi_hermite_eval} is correct for all inputs on $n-1$ variables. Furthermore, suppose that Algorithm~\ref{alg:multi_hermite_eval} is called on a finite initial segment $I\subseteq\N^n$ and the vector $F\dashv(N_{\vec{i}})_{\vec{i}\in I}=(f_{\vec{i}})_{\vec{i}\in I}$ for some polynomial $F\in\F_q[\vec{X}]_I$. Then the subvectors $(f_{(i_1,\dotsc,i_{n-1},i_n)})_{(i_1,\dotsc,i_{n-1})\in\lambda(I,i_n)}$ for $i_n\in\rho(I,\vec{0}_{n-1})$ are modified independently of one another during the first $n-1$ iterations of the outer loop of the algorithm. Indeed, during the first $n-1$ iterations of the outer loop, the univariate algorithm is only ever called on a subvector of one of these subvectors. For $\ell\in\{1,\dotsc,n-1\}$, the family of sets
\begin{equation*}
	\left\{
		\left(j_1,\dotsc,j_{\ell-1},i_{\ell+1},\dotsc,i_{n-1},i_n\right)
		\mid
		\left(j_1,\dotsc,j_{\ell-1},i_{\ell+1},\dotsc,i_{n-1}\right)
		\in\pi_\ell\left(\lambda(I,i_n)\right)
	\right\}
\end{equation*}
for $i_n\in\rho(I,\vec{0}_{n-1})$ form a partition of $\pi_\ell(I)$. Moreover, for $i_n\in\rho(I,\vec{0}_{n-1})$ and $\vec{k}'=(j_1,\dotsc,j_{\ell-1},i_{\ell+1},\dotsc,i_{n-1})\in\pi_\ell\left(\lambda(I,i_n)\right)$, we have
\begin{equation*}
	\mu_\ell\left(
		I,
		\left(j_1,\dotsc,j_{\ell-1},i_{\ell+1},\dotsc,i_{n-1},i_n\right)
	\right)
	=\mu_\ell\left(
	\lambda\left(I,i_n\right),\vec{k}'\right).
\end{equation*}
Thus, performing the first $n-1$ iterations of the outer loop is equivalent to recursively calling the algorithm on the initial segment $\lambda(I,i_n)\subseteq\N^{n-1}$ and the subvector $(f_{(i_1,\dotsc,i_{n-1},i_n)})_{(i_1,\dotsc,i_{n-1})\in\lambda(I,i_n)}$ for each $i_n\in\rho(I,\vec{0}_{n-1})$. Initially, we have
\begin{equation*}
	\left(f_{(i_1,\dotsc,i_{n-1},i_n)}\right)_{(i_1,\dotsc,i_{n-1})\in\lambda(I,i_n)}
	=F_{i_n}\dashv
	\left(N_{(i_1,\dotsc,i_{n-1})}\right)_{(i_1,\dotsc,i_{n-1})\in\lambda(I,i_n)}
\end{equation*}
for $i_n\in\rho(I,\vec{0}_{n-1})$. Therefore, the induction hypothesis implies that after $n-1$ iterations of the outer loop have been performed, the input vector is equal to $(\evH(F_{i_n},(j_1,\dotsc,j_{n-1})))_{(j_1,\dotsc,j_{n-1},i_n)\in I}$. It follows that the last iteration of the outer loop calls the univariate algorithm on each of the vectors
\begin{equation*}
	\left(f_{(j_1,\dotsc,j_{n-1},i_n)}\right)_{i_n\in\mu_n(I,\vec{k})}
	=\left(
		\evH\left(F_{i_n},(j_1,\dotsc,j_{n-1})\right)
	\right)_{i_n\in\rho(I,(j_1,\dotsc,j_{n-1}))}
\end{equation*}
for $\vec{k}=(j_1,\dotsc,j_{n-1})\in\pi_n(I)=\lambda(I,0)$. Hence, Lemma~\ref{lem:recursion} implies that Algorithm~\ref{alg:multi_hermite_eval} returns the vector $(E(F,\vec{j}))_{\vec{j}\in I}$. That is, the algorithm is correct for the inputs $I$ and $F\dashv(N_{\vec{i}})_{\vec{i}\in I}$. Thus, the lemma follows by induction.
\end{proof}
\end{samepage}

It is clear that the complexity of each multivariate algorithm is determined by the complexity of the corresponding univariate algorithm. We capture the nature of this dependency in the next two lemmas.

\begin{lemma}\label{lem:multi_hermite_eval} Suppose that for some function $\mathsf{E}:\N\setminus\{0\}\rightarrow\RR$ the univariate Hermite evaluation algorithm used in Algorithm~\ref{alg:multi_hermite_eval} performs at most $\mathsf{E}(s)$ operations in $\F_q$ when given the initial segment $[s]$ as an input, and that $\mathsf{E}(s)/s$ is a nondecreasing function of $s$. Then given an input such that $I\subseteq[s_1]\times\dotsb\times[s_n]$ for positive integers $s_1,\dotsc,s_n$, Algorithm~\ref{alg:multi_hermite_eval} performs at most
\begin{equation*}
	\left(
		\frac{\mathsf{E}(s_1)}{s_1}
		+\dotsb+
		\frac{\mathsf{E}(s_n)}{s_n}
	\right)
	\left|I\right|
\end{equation*}
operations in $\F_q$.
\end{lemma}
\begin{proof} If the univariate Hermite evaluation algorithm has complexity given by such a function $\mathsf{E}:\N\rightarrow\RR$, then Algorithm~\ref{alg:multi_hermite_eval} performs at most
\begin{equation*}
	\sum^n_{\ell=1}
	\sum_{\vec{k}\in\pi_\ell(I)}\mathsf{E}(\left|\mu_\ell(I,\vec{k})\right|)
	=\sum^n_{\ell=1}
	\sum_{\vec{k}\in\pi_\ell(I)}
	\frac{\mathsf{E}(\left|\mu_\ell(I,\vec{k})\right|)}
	{\left|\mu_\ell(I,\vec{k})\right|}
	\left|\mu_\ell(I,\vec{k})\right|
\end{equation*}
operations in $\F_q$. It follows that if $I\subseteq[s_1]\times\dotsb\times[s_n]$ for positive integers $s_1,\dotsc,s_n$, and thus $\mu_\ell(I,\vec{k})\subseteq[s_\ell]$ for $\ell\in\{1,\dotsc,n\}$ and $\vec{k}\in\pi_\ell(I)$, then the algorithm performs at most
\begin{equation*}
	\sum^n_{\ell=1}
	\frac{\mathsf{E}(s_\ell)}{s_\ell}
	\sum_{\vec{k}\in\pi_\ell(I)}\left|\mu_\ell(I,\vec{k})\right|
	=\sum^n_{\ell=1}\frac{\mathsf{E}(s_\ell)}{s_\ell}\left|I\right|
\end{equation*}
operations in $\F_q$.
\end{proof}

\begin{lemma}\label{lem:multi_hermite_interp} Suppose that for some function $\mathsf{I}:\N\setminus\{0\}\rightarrow\RR$ the univariate Hermite interpolation algorithm used in Algorithm~\ref{alg:multi_hermite_interp} performs at most $\mathsf{I}(s)$ operations in $\F_q$ when given the initial segment $[s]$ as an input, and that $\mathsf{I}(s)/s$ is a nondecreasing function of $s$. Then given an input such that $I\subseteq [s_1]\times\dotsb\times[s_n]$ for positive integers $s_1,\dotsc,s_n$, Algorithm~\ref{alg:multi_hermite_interp} performs at most
\begin{equation*}
	\left(
		\frac{\mathsf{I}(s_1)}{s_1}
		+\dotsb+
		\frac{\mathsf{I}(s_n)}{s_n}
	\right)
	\left|I\right|
\end{equation*}
operations in $\F_q$.
\end{lemma}

We omit the proof of Lemma~\ref{lem:multi_hermite_interp} since it uses identical arguments to those of Lemma~\ref{lem:multi_hermite_eval}. Combining the two lemmas with Lemmas~\ref{lem:univariate_hermite} and~\ref{lem:univariate_conversion} then completes the proof of Theorem~\ref{thm:multivariate_hermite}. We note that Lemmas~\ref{lem:multi_hermite_eval} and~\ref{lem:multi_hermite_interp}, and thus Theorem~\ref{thm:multivariate_hermite}, do not account for the cost of computing the sets $\pi_\ell(I)$ and $\mu_\ell(I,\vec{k})$ during the algorithm. For the initial segments that are used in the encoding algorithms, we have simple explicit formulae that allow the sets to be computed with low complexity. The general problem is not considered here.

As for the complexities of the multivariate algorithms, their space requirements are largely determined by those of the univariate algorithms. The amount of auxiliary space used by either multivariate algorithm, i.e., storage in addition to the input array, is equal to that of the index manipulations plus the maximum amount of auxiliary space used by the corresponding univariate algorithm over all calls to it. Therefore, if the univariate algorithm is a true in-place algorithm, in the sense that it uses only $\bigO(1)$ auxiliary space, and the index manipulations also require only $\bigO(1)$ auxiliary space, then the multivariate algorithm enjoys the same auxiliary space bound.

\section{Encoding algorithm for low-rate codes}\label{sec:low_rate}

In this section, we present the first of our fast systematic encoding algorithms for multiplicity codes. Although, the algorithm is suitable for multiplicity codes of all rates, we somewhat falsely refer to it as an encoding algorithm for low-rate codes since the encoding algorithm of Section~\ref{sec:high_rate} is faster for codes with sufficiently high rates. Recall that our goal is to efficiently evaluate the encoding function $\mathrm{enc}^s_d$ defined in Section~\ref{sec:encoding}. The algorithm of this section achieves this goal by using the fast Hermite interpolation and evaluation algorithms of Section~\ref{sec:hermite} to successively evaluate its constituent maps $\eval^{-1}_{\Ical_{d,n}}$ and $\eval^s_d$.

We use the map $\kappa_n:[q]^n\times\N^n\rightarrow\N^n$, given by $(\vec{i},\vec{s})\mapsto\vec{i}+\vec{s}q$, to translate the encoding problem into the language of Section~\ref{sec:hermite}. To this end, we let $\I_{d,n}$ denote the $\kappa_n$-image of the information set $\Ical_{d,n}$ defined in Theorem~\ref{thm:info_set}. Then we have
\begin{equation*}
	\I_{d,n}
	=\{\vec{i}\in\N^n\mid\left|\vec{i}\right|\leq d\}
	\quad\text{for $d\in\N$}.
\end{equation*}
For notational convenience, we extend this definition to $d\in\Z$, by defining $\I_{d,n}$ to be the empty set for $d<0$. For nonzero $s\in\N$, we define $\C_{s,n}=\kappa_n([q]^n\times S_{s,n})$. Finally, for $d,s\in\Z$ such that $d<sq$ and $s>0$, we define $\R_{d,s,n}=\C_{s,n}\setminus\I_{d,n}$. With this notation, a message of a multiplicity code $\Mult^s_d$ is written as a vector $m=(m_{\vec{j}})_{\vec{j}\in\I_{d,n}}$. Its systematic encoding is then equal to
\begin{equation*}
	\mathrm{enc}^s_d(m)
	=\left(
	\left(\evH(F,\vec{j}+\vec{t}q)\right)_{\vec{t}\in S_{s,n}}
	\right)_{\vec{j}\in[q]^n}
\end{equation*}
where $F$ is the unique polynomial in $\F_q[\vec{X}]_d$ such that $\evH(F,\vec{j})=m_{\vec{j}}$ for $\vec{j}\in\I_{d,n}$. It follows that it is sufficient to consider the problem of computing the vector $(\evH(F,\vec{j}))_{\vec{j}\in\C_{s,n}}$ when given $m$. In fact, we need only compute $(\evH(F,\vec{j}))_{\vec{j}\in\R_{d,s,n}}$ since the remaining entries are present in the message $m$ to begin with.

For $d,s\in\N$ such that $d<sq$, we have $\F_q[\vec{X}]_d=\F_q[\vec{X}]_{I_{d,n}}\subset\F_q[\vec{X}]_{C_{s,n}}$. Therefore, as noted in Section~\ref{sec:hermite_definition}, computing the polynomial $F$ that corresponds to a message $m$ of the multiplicity code $\Mult^s_d$ (i.e., computing $\eval^{-1}_{\Ical_{d,n}}(m)$) is an instance of the Hermite interpolation problem with initial segment $I=\I_{d,n}$, while computing the vector $(\evH(F,\vec{j}))_{\vec{j}\in\C_{s,n}}$ (i.e., computing \emph{the entries of} $\eval^s_d(F)$) is an instance of Hermite evaluation problem with initial segment $I=\C_{s,n}$. Applying the algorithms of Section~\ref{sec:hermite} to these instances of the interpolation and evaluation problems yields our first systematic encoding algorithm, presented in Algorithm~\ref{alg:low_rate}, and the complexity bound of Theorem~\ref{thm:encoders}.

\begin{algorithm}[!ht]
	\caption{Systematic encoding for low-rate codes}\label{alg:low_rate}
	\begin{algorithmic}[1]
		\Require Nonnegative integers $d$ and $s$ such that $d<sq$; and a vector $(c_{\vec{j}})_{\vec{j}\in\C_{s,n}}$ such that $m=(c_{\vec{j}})_{\vec{j}\in\I_{d,n}}$ is a message of $\Mult^s_d$, and $c_{\vec{j}}=0$ for $\vec{j}\in\R_{d,s,n}$.
		
		\Ensure The vector $(c_{\vec{j}})_{\vec{j}\in\C_{s,n}}$ equal to $(\evH(F,\vec{j}))_{\vec{j}\in\C_{s,n}}$ for the polynomial $F\in\F_q[\vec{X}]_d$ that corresponds to the message $m$.
		
		\State Call Algorithm~\ref{alg:multi_hermite_interp} on $\I_{d,n}$ and $(c_{\vec{j}})_{\vec{j}\in\I_{d,n}}$.
		
		\State Call Algorithm~\ref{alg:multi_hermite_eval} on $\C_{s,n}$ and $(c_{\vec{j}})_{\vec{j}\in\C_{s,n}}$.
	\end{algorithmic}
\end{algorithm}

\begin{theorem}\label{thm:encoders} Given a message vector $m$ of a multiplicity code $\Mult^s_d$, its systematic encoding $\mathrm{enc}^s_d(m)$ can be computed in
\begin{equation*}
	\bigO\left(
		\frac{\mult(sq)\log sq}{sq}\left|\C_{s,n}\right|n
	\right)
\end{equation*}
operations in $\F_q$.
\end{theorem}
\begin{proof} Taking the univariate algorithms used by Algorithms~\ref{alg:multi_hermite_eval} and~\ref{alg:multi_hermite_interp} to be the corresponding algorithms of Chin, as modified in Section~\ref{sec:univariate}, Theorem~\ref{thm:multivariate_hermite} implies that Algorithm~\ref{alg:low_rate} performs
\begin{equation}\label{eqn:low_rate}
	\bigO\left(
		\frac{\mult(d+1)\log(d+1)}{d+1}\left|\I_{d,n}\right|n
		+\frac{\mult(sq)\log sq}{sq}\left|\C_{s,n}\right|n
	\right)
\end{equation}
operations in $\F_q$. As the parameters $d$ and $s$ satisfy the inequality $d<sq$, and thus $\I_{d,n}\subset\C_{s,n}$, the second term of the bound dominates.
\end{proof}

By taking $\mult(k)$ to be in $\bigO(k\log k\log\log k)$, it follows from Theorem~\ref{thm:encoders} that systematic encoding for $\Mult^s_d$ can be performed in $\bigO(\left|\C_{s,n}\right|n\log^2(sq)\log\log(sq))$ operations in $\F_q$, matching the quasi-linear bound stated in the introduction. While we have only bounded the number of field operations performed by the encoding algorithm, the cost of the index manipulations performed by Algorithms~\ref{alg:multi_hermite_eval} and~\ref{alg:multi_hermite_interp} during encoding is low in practice. Indeed, for $\ell\in\{1,\dotsc,n\}$, we have the simple explicit formulae $\pi_\ell(\I_{d,n})=\I_{d,n-1}$ and $\mu_\ell(\I_{d,n},\vec{k})=\I_{d-\left|\vec{k}\right|,1}$ for $\vec{k}\in\I_{d,n-1}$. Similarly, $\pi_\ell(\C_{s,n})=\C_{s,n-1}$ and $\mu_\ell(\C_{s,n},\vec{k})=\C_{s-\left|\vec{k}\bdiv{q}\right|,1}$ for $\vec{k}\in\C_{s,n-1}$.

While the encoding algorithm has quasi-optimal asymptotic complexity, it is clear that it performs more operations than is necessary. This excess is most apparent in the evaluation step of the algorithm, which recomputes the entries of original input message. We address this problem in the next section by describing some modifications to the encoding algorithm that can be used to eliminate unnecessary computations. Each modification requires modifications to be made to one or both of the univariate algorithms. As we are being non-committal about our choice of these algorithms, we only describe how the behaviour of univariate algorithms should be changed, rather than describing how to obtain the desired behaviour.

\subsection{Practical improvements}\label{sec:improvements}

Our first modification occurs at the interface of Algorithms~\ref{alg:multi_hermite_eval} and~\ref{alg:multi_hermite_interp}. Suppose that, as occurs for the algorithms of Section~\ref{sec:univariate}, the univariate interpolation algorithm performs monomial to Newton basis conversion as its last step, and the univariate evaluation algorithm performs the inverse conversion as its first step. Then the conversions performed during the last iteration of the interpolation algorithm cancel with those performed during the first iteration of the evaluation algorithm. Consequently, these basis conversions can be avoided altogether, saving $\Omega(\left|\I_{d,n}\right|)$ operations.

For our second improvement, we modify the evaluation step of the encoding algorithm to take advantage of the fact that the polynomial being evaluated has support contained in $\I_{d,n}$, a proper, and possibly much smaller, subset of the initial segment $I=\C_{s,n}$ for which we apply Algorithm~\ref{alg:multi_hermite_eval}. In the univariate case, we need only modify the algorithm to take advantage of the fact that the polynomial has degree at most $d$ rather than at most $sq-1$. And it is straightforward to modify the algorithm of Section~\ref{sec:univariate} accordingly. The following lemma allows us to extend the modified univariate algorithm to the multivariate case.

\begin{lemma}\label{lem:low_rate_zero_entries} If the inputs of Algorithm~\ref{alg:multi_hermite_eval} satisfy $I\supset\I_{d,n}$ and $f_{\vec{i}}=0$ for $\vec{i}\in I\setminus\I_{d,n}$, for some $d\in\N$, then at the beginning of the $\ell$th iteration of the outer loop of the algorithm, for $\ell\in\{1,\dotsc,n\}$, we have $f_{\vec{i}}=0$ for all $\vec{i}=(i_1,\dotsc,i_n)\in I$ such that $i_\ell+\dotsb+i_n>d$.
\end{lemma}
\begin{proof} We prove the lemma by induction on $\ell$. The statement holds trivially for the first iteration. Therefore, suppose that at the beginning of the $\ell$th iteration of the outer loop, for some $1\leq\ell<n$, we have $f_{\vec{i}}=0$ for all $\vec{i}=(i_1,\dotsc,i_n)\in I$ such that $i_\ell+\dotsb+i_n>d$. Then for all $\vec{k}=(j_1,\dotsc,j_{\ell-1},i_{\ell+1},\dotsc,i_n)\in\pi_\ell(I)$ such that $i_{\ell+1}+\dotsb+i_n>d$, the subvector $(f_{(j_1,\dotsc,j_{\ell-1},i_\ell,i_{\ell+1},\dotsc,i_n)})_{i_\ell\in\mu_\ell(I,\vec{j})}$ contains all zeros and is consequently unchanged by the call to the univariate algorithm. As the sets
\begin{equation*}
	\left\{
		(j_1,\dotsc,j_{\ell-1},i_\ell,i_{\ell+1},\dotsc,i_n)
		\mid
		i_\ell\in\mu_\ell(I,\vec{k})
	\right\}
\end{equation*}
for $\vec{k}=(j_1,\dotsc,j_{\ell-1},i_{\ell+1},\dotsc,i_n)\in\pi_\ell(I)$ form a partition of $I$, it follows that at the beginning of the next iteration we have $f_{\vec{i}}=0$ for all $\vec{i}=(i_1,\dotsc,i_n)\in I$ such that $i_{\ell+1}+\dotsb+i_n>d$.
\end{proof}

It follows from Lemma~\ref{lem:low_rate_zero_entries} that during the evaluation step of the encoding algorithm, the inner loop of Algorithm~\ref{alg:multi_hermite_eval} need only be performed for $\vec{k}$ such that $i_{\ell+1}+\dotsb+i_n\leq d$. Moreover, for each such $\vec{k}$, the univariate algorithm evaluates a polynomial of degree at most $d-i_{\ell+1}-\dotsb-i_n$. Consequently, it is straightforward to extend the modified univariate algorithm to the multivariate case. To give some indication of the number of operations saved by this modification, we observe that the number of zero entries described by Lemma~\ref{lem:low_rate_zero_entries} over all iterations is equal to
\begin{equation*}
	\left|\R_{d,s,n}\right|
	+\sum^n_{\ell=2}
	\sum_{\vec{i}\in\R_{d,s,n-\ell+1}}
	\left|\C_{s-\left|\vec{i}\bdiv{q}\right|,\ell-1}\right|.
\end{equation*}
Hence, the modification saves the most operations when the rate $\left|\I_{d,n}\right|/\left|\C_{s,n}\right|$ of the code is low.

For the final modification, we stop the evaluation step of the encoding algorithm from recomputing the input message, saving $\Omega(\left|\I_{d,n}\right|)$ operations. These entries of the input and output are indexed by the information set $\I_{d,n}$. Consequently, during the last iteration of the outer loop of Algorithm~\ref{alg:multi_hermite_eval}, where we call the univariate algorithm on $(f_{(j_1,\dotsc,j_{n-1},i_n)})_{i_n\in\C_{s-\left|\vec{k}\bdiv{q}\right|,n}}$ for each $\vec{k}=(j_1,\dotsc,j_{n-1})\in\C_{s,n-1}$, the univariate algorithm need only return correct values in those entries with $i_n>d-j_1-\dotsb-j_{n-1}$. The entries indexed by $\R_{d,s,n}$ will then be correct at the end of the algorithm, while the remaining entries will contain meaningless values. Therefore, if the modification can be implement for the univariate case, then it readily extends to the multivariate case.

\section{Encoding algorithm for high-rate codes}\label{sec:high_rate}

Let $m=(m_{\vec{j}})_{\vec{j}\in\I_{d,n}}$ be a message vector of a multiplicity code $\Mult^s_d$. Then Lemma~\ref{lem:existence_uniqueness} implies that there exists a unique polynomial $F_\C\in\F_q[\vec{X}]_{\C_{s,n}}$ such that
\begin{equation*}
	\evH(F_\C,\vec{j})
	=\begin{cases}
		m_{\vec{j}} & \text{if $\vec{j}\in\I_{d,n}$},   \\
		0           & \text{if $\vec{j}\in\R_{d,s,n}$}.
	\end{cases}
\end{equation*}
Let $F_\C\dashv(N_{\vec{i}})_{\vec{i}\in\C_{s,n}}=(f_{\vec{i}})_{\vec{i}\in\C_{s,n}}$, and define polynomials $F_\I=\sum_{\vec{i}\in\I_{d,n}}f_{\vec{i}}N_{\vec{i}}$ and $F_\R=-\sum_{\vec{i}\in\R_{d,s,n}}f_{\vec{i}}N_{\vec{i}}$. Then $F_\C=F_\I-F_\R$. As $\I_{d,n}$ is an initial segment that is disjoint with $\R_{d,s,n}$, we have $\vec{j}\ngeq\vec{i}$ for $\vec{j}\in\I_{d,n}$ and $\vec{i}\in\R_{d,s,n}$. Thus, Lemma~\ref{lem:basis_property} implies that
\begin{equation*}
	\evH(F_\I,\vec{j})
	=\evH(F_\C,\vec{j})+\evH(F_\R,\vec{j})
	=\begin{cases}
		m_{\vec{j}}        & \text{if $\vec{j}\in\I_{d,n}$},   \\
		\evH(F_\R,\vec{j}) & \text{if $\vec{j}\in\R_{d,s,n}$}.
	\end{cases}
\end{equation*}
It follows that $F_\I\in\F_q[\vec{X}]_{\I_{d,n}}=\F_q[\vec{X}]_d$ is the polynomial that corresponds to the message~$m$. Moreover, as $\evH(F_\I,\vec{j})$ and $\evH(F_\R,\vec{j})$ agree for $\vec{j}\in\R_{d,s,n}$, the polynomial $F_\R$ may be used to compute the unknown entries of the systematic encoding $\eval^s_d(F_\I)$ of $m$. In this section, we show that when the rate $\left|\I_{d,n}\right|/\left|\C_{s,n}\right|$ of $\Mult^s_d$ is sufficiently close to one, so that $F_\R$ has much fewer nonzero coefficients on the Newton basis than $F_\I$, the unknown entries of the systematic encoding can be computed more efficiently by using $F_\R$ in place of $F_\I$. When this gain is sufficient to compensate for the extra cost of computing $F_\R$ (for which we first compute $F_\C$), when compared to that of directly computing $F_\I$, we also gain an advantage over the encoding algorithm of Section~\ref{sec:low_rate}.

To compute the values $\evH(F_\R,\vec{j})$ for $\vec{j}\in\R_{d,s,n}$, we use Algorithm~\ref{alg:multi_hermite_eval} with $I=\C_{s,n}$ as our starting point. Then following an approach similar to that used in Section~\ref{sec:improvements}, we eliminate unnecessary operations from the algorithm by taking advantage of the fact that $F_\R\dashv(N_{\vec{i}})_{\vec{i}\in\C_{s,n}}$ has zeros in those entries with indices in $\I_{d,n}$. Once again we find that the multivariate case follows readily from the univariate case. Consequently, we begin this section by considering the univariate problem.

\subsection{Univariate algorithm}\label{sec:uni_hermite_eval_R}

Let $s$ be a positive integer and $F\in\F_q[X]$. Then the definition of the Hasse derivative implies that
\begin{equation*}
	F(X+\alpha_j)\bmod{X^s}
	=\sum^{s-1}_{t=0}H(F,t)(\alpha_j)X^t
	=\sum^{s-1}_{t=0}\evH(F,j+tq)X^t
	\quad\text{for $j\in[q]$}.
\end{equation*}
Thus, computing the values $\evH(F,j)$ for $j\in\C_{s,1}$ is equivalent to computing the polynomials $F(X+\alpha_j)\bmod{X^s}$ on the monomial basis for $j\in[q]$. Suppose now that for some nonnegative integer $d<sq-1$, the polynomial $F$ is of the form $\sum_{i\in\R_{d,s,1}}f_iN_i$ with coefficients $f_i\in\F_q$. Then $N_{d+1}$ divides~$F$, allowing us to reduce the problem of computing the values $\evH(F,j)$ for $j\in\C_{s,1}$ to the lower degree problems of computing $\evH(N_{d+1},j)$ and $\evH(F/N_{d+1},j)$ for $j\in\C_{s,1}$, with polynomial multiplications modulo~$X^s$ to combine the results. We improve upon this approach by replacing $N_{d+1}$ by its largest factor that is invariant under Taylor shifts.

Let $r=\floor{(d+1)/q}$. Then $N_{rq}$ divides $F$ since $rq\leq d+1$, and
\begin{equation}\label{eqn:N_rq}
	N_{rq}
	=\prod^{rq-1}_{j=0}\left(X-\alpha_{j\bmod{q}}\right)
	=\prod^{q-1}_{j=0}\left(X-\alpha_j\right)^r
	=\left(X^q-X\right)^r.
\end{equation}
Therefore, if we let $Q=F/N_{rq}\in\F_q[X]$, then
\begin{equation*}
	F(X+\alpha_j)\bmod{X^s}
	=X^r\left((X^{q-1}-1)^rQ(X+\alpha_j)\bmod{X^{s-r}}\right)
	\quad\text{for $j\in[q]$}.
\end{equation*}
Hence, we can compute the values $\evH(F,j)$ for $j\in[sq]\setminus[rq]\supseteq\R_{d,s,1}$ by first computing the polynomials $Q(X+\alpha_j)\bmod{X^{s-r}}$ on the monomial basis for $j\in[q]$, for which we can use Chin's Hermite evaluation algorithm, then multiplying each shifted polynomial by $(X^{q-1}-1)^r\bmod{X^{s-r}}$. Applying this approach, we obtain Algorithm~\ref{alg:uni_hermite_eval_R}. We allow the input $d$ of the algorithm to be negative, in which case $\R_{d,s,1}=\C_{s,1}$, in order to simplify the description of the multivariate algorithm in the next section.

\begin{algorithm}[!ht]
	\caption{Univariate Hermite evaluation over $\R_{d,s,1}$}\label{alg:uni_hermite_eval_R}
	\begin{algorithmic}[1]
		\Require Integers $d$ and $s$ such that $d+1<sq$ and $s>0$; the vector $F\dashv(N_i)_{i\in\R_{d,s,1}}=(f_i)_{i\in\R_{d,s,1}}$ for some $F\in\F_q[X]$ of the form $\sum_{i\in\R_{d,s,1}}f_iN_i$; and the polynomial $(X^{q-1}-1)^r\bmod{X^{s-r}}$ for $r=\max(\floor{(d+1)/q},0)$, written on the monomial basis.
		\Ensure The vector $(f_i)_{i\in\R_{d,s,1}}$ equal to $(\evH(F,j))_{j\in\R_{d,s,1}}$.
		\State Set $r=\max(\floor{(d+1)/q},0)$.
		\State\label{step:convert}Use the algorithm of Gerhard~\cite{gerhard2000} to convert $Q=F/N_{rq}$ to the monomial basis.
		\State\label{step:hermite_eval}Use the Hermite evaluation algorithm of Chin~\cite{chin1976} with the modifications described in Section~\ref{sec:univariate} to compute the coefficients of the polynomials $Q(X+\alpha_j)\bmod{X^{s-r}}$ for $j\in[q]$ on the monomial basis.
		\For{$j\in[q]$}\label{step:multiply}
			\Statew{2}{Compute $((X^{q-1}-1)^r\bmod{X^{s-r}})(Q(X+\alpha_j)\bmod{X^{s-r}})$ and set $f_{j+tq}$ equal to the coefficient of $X^{t-r}$ in the resulting polynomial for $t=\max(\ceil{(d+1-j)/q},0),\dotsc,s-1$.}
		\EndFor\label{step:multiply-end}
	\end{algorithmic}
\end{algorithm}

\begin{lemma}\label{lem:uni_hermite_eval_R} Algorithm~\ref{alg:uni_hermite_eval_R} performs $\bigO\left(\mult((s-r)q)\log((s-r)q)\right)$ operations in $\F_q$, where  $r=\max(\floor{(d+1)/q},0)$.
\end{lemma}
\begin{proof} Equation~\eqref{eqn:N_rq} implies that $N_{rq+i}=N_{rq}N_i$ for $i,r\in\N$. Consequently, in Line~\ref{step:convert} of the algorithm, the coefficient vector of $Q$ on the Newton basis can be read directly from the coefficient vector of $F$. As $\deg Q<(s-r)q$, Lemma~\ref{lem:univariate_conversion} therefore implies that Line~\ref{step:convert} performs $\bigO\left(\mult\left((s-r)q\right)\log\left((s-r)q\right)\right)$ operations in~$\F_q$. Similarly, Lemma~\ref{lem:univariate_hermite} implies that Line~\ref{step:hermite_eval} performs $\bigO\left(\mult\left((s-r)q\right)\log\left((s-r)q\right)\right)$ operations in~$\F_q$. Finally, Lines~\ref{step:multiply}--\ref{step:multiply-end} perform $q$ multiplications of polynomials with degree less than $s-r$, requiring at most $q\mult(s-r)\leq\mult((s-r)q)$ operations in $\F_q$. Hence, Algorithm~\ref{alg:uni_hermite_eval_R} performs $\bigO\left(\mult((s-r)q)\log((s-r)q)\right)$ operations in $\F_q$.
\end{proof}

We have included the polynomial $(X^{q-1}-1)^r\bmod{X^{s-r}}$ as an input to Algorithm~\ref{alg:uni_hermite_eval_R} for the benefit of the multivariate algorithm of the next section, which is able to reuse these inputs for multiple calls to the algorithm. For an instance of the univariate problem, this input can be computed in $\bigO(\mult(\ceil{(s-r)/(q-1)})\log r)$ operations in $\F_q$ by the square and multiply algorithm for exponentiation. Alternatively, one can use the binomial theorem and Lucas' lemma~\cite[p.~230]{lucas1878} (see also~\cite{fine1947}).

\subsection{Multivariate algorithm}\label{sec:multi_hermite_eval_R}

Recall that the polynomial $F_\R$ defined in Section~\ref{sec:high_rate} has support on the monomial basis that is contained in $\C_{s,n}$, while its coefficient vector $F_\R\dashv(N_{\vec{i}})_{\vec{i}\in\C_{s,n}}$ has zeros in those entries with indices in $\I_{d,n}$ for some $d<sq$. The following lemma implies that if Algorithm~\ref{alg:multi_hermite_eval} is called on $\C_{s,n}$ and the coefficient vector of $F_\R$, then the zeros in the entries indexed by $\I_{d,n}$ persist throughout the algorithm.

\begin{lemma}\label{lem:zero_entries} If the inputs of Algorithm~\ref{alg:multi_hermite_eval} satisfy $f_{\vec{i}}=0$ for $\vec{i}\in I\cap\I_{d,n}$, for some $d\in\N$, then at the beginning of the $\ell$th iteration of the outer loop of the algorithm, for $\ell\in\{1,\dotsc,n\}$, we have $f_{\vec{i}}=0$ for $\vec{i}\in I\cap\I_{d,n}$.
\end{lemma}
\begin{proof} We prove the lemma by induction on $\ell$. The statement holds trivially for the first iteration. Therefore, suppose that for some $d\in\N$ and $\ell\in\{1,\dotsc,n-1\}$ we have $f_{\vec{i}}=0$ for $\vec{i}\in I\cap\I_{d,n}$ at the beginning of the $\ell$th iteration of the outer loop. Let $\vec{k}=(j_1,\dotsc,j_{\ell-1},i_{\ell+1},\dotsc,i_n)\in\pi_\ell(I)$. Then $f_{(j_1,\dotsc,j_{\ell-1},i_\ell,i_{\ell+1},\dotsc,i_n)}=0$ for $i_\ell\leq d-\left|\vec{k}\right|$, and Lemma~\ref{lem:basis_property} implies that
\begin{equation*}	
	\evH\left(
		\sum_{i_\ell\in\mu_\ell(I,\vec{k})}
		f_{(j_1,\dotsc,j_{\ell-1},i_\ell,i_{\ell+1},\dotsc,i_n)}N_{i_\ell}
		,j_\ell
	\right)
	=\evH\left(
		\sum^{j_\ell}_{i_\ell=0}
		f_{(j_1,\dotsc,j_{\ell-1},i_\ell,i_{\ell+1},\dotsc,i_n)}N_{i_\ell}
		,j_\ell
	\right)
\end{equation*}
for $j_\ell\in\mu_\ell(I,\vec{k})$. Thus, the entries of $(f_{(j_1,\dotsc,j_{\ell-1},i_\ell,i_{\ell+1},\dotsc,i_n)})_{i_\ell\in\mu_\ell(I,\vec{k})}$ with $i_\ell\leq d-\left|\vec{k}\right|$ are still zero after the univariate Hermite evaluation algorithm has been called on the vector in Line~\ref{step:eval}. Hence, $f_{\vec{i}}=0$ for $\vec{i}\in I\cap\I_{d,n}$ at the end of the $\ell$th iteration of the outer loop.
\end{proof}

Let $d,s\in\N$ such that $d<sq$, and $F\in\F_q[\vec{X}]_{\C_{s,n}}$ such that its coefficient vector $F\dashv(N_{\vec{i}})_{\vec{i}\in\C_{s,n}}=(f_{\vec{i}})_{\vec{i}\in\C_{s,n}}$ satisfies $f_{\vec{i}}=0$ for $\vec{i}\in\I_{d,n}$. Then it follows from Lemma~\ref{lem:zero_entries} that if Algorithm~\ref{alg:multi_hermite_eval} is called on $\C_{s,n}$ and $F\dashv(N_{\vec{i}})_{\vec{i}\in\C_{s,n}}$, then each time Line~\ref{step:eval} of the algorithm is executed, the vector 
\begin{equation*}
	\left(
		f_{(j_1,\dotsc,j_{\ell-1},i_\ell,i_{\ell+1},\dotsc,i_n)}
	\right)_{i_\ell\in\mu_\ell(\C_{s,n},\vec{k})}
	=
	\left(
		f_{(j_1,\dotsc,j_{\ell-1},i_\ell,i_{\ell+1},\dotsc,i_n)}
	\right)_{i_\ell\in\C_{s-\left|\vec{k}\bdiv{q}\right|,1}}
\end{equation*}
has zeros in those entries with $i_\ell\in\I_{d-\left|\vec{k}\right|,1}$. We can take advantage of these zero entries by modifying Line~\ref{step:eval} so that Algorithm~\ref{alg:uni_hermite_eval_R} is called on the vector $(f_{(j_1,\dotsc,j_{\ell-1},i_\ell,i_{\ell+1},\dotsc,i_n)})_{i_\ell\in\R_{d-\left|\vec{k}\right|,s-\left|\vec{k}\bdiv{q}\right|,1}}$. This modification requires that Algorithm~\ref{alg:uni_hermite_eval_R} is provided with the polynomial
\begin{equation*}
	(X^{q-1}-1)^r\bmod{X^{s-\left|\vec{k}\bdiv{q}\right|-r}}
	\quad\text{for $r=\max(\floor{(d-\left|\vec{k}\right|+1)/q},0)$}.
\end{equation*}
Therefore, along with the modification to Line~\ref{step:eval} of the algorithm, we can introduce a precomputation step to the algorithm where the polynomials
\begin{equation*}
	(X^{q-1}-1)^r\bmod{X^{s-r}}
	\quad\text{for $r=0,\dotsc,\floor{(d+1)/q}$}
\end{equation*}
are computed on the monomial basis. Then each call to Algorithm~\ref{alg:uni_hermite_eval_R} only requires that one of these polynomials be trivially reduced modulo some power of $X$. By making these modifications to Algorithm~\ref{alg:multi_hermite_eval}, we obtain Algorithm~\ref{alg:multi_hermite_eval_R}.

\begin{algorithm}[!ht]
	\caption{Multivariate Hermite evaluation over $\R_{d,s,n}$}\label{alg:multi_hermite_eval_R}
	\begin{algorithmic}[1]
		\Require Integers $d$ and $s$ such that $0\leq d<sq$; and the vector $F\dashv(N_{\vec{i}})_{\vec{i}\in\R_{d,s,n}}=(f_{\vec{i}})_{\vec{i}\in\R_{d,s,n}}$ for some $F\in\F_q[\vec{X}]$ of the form $\sum_{\vec{i}\in\R_{d,s,n}}f_{\vec{i}}N_{\vec{i}}$ with $n\geq 2$.
		
		\Ensure The vector $(f_{\vec{i}})_{\vec{i}\in\R_{d,s,n}}$ equal to $(\evH(F,\vec{j}))_{\vec{j}\in\R_{d,s,n}}$.
		
		\State Set $U_0=1$.
		
		\For{$r=1,\dotsc,\floor{(d+1)/q}$}\label{step:pre-comp}
			\State Compute $U_r=(X^{q-1}-1)U_{r-1}\bmod{X^{s-r}}$ on the monomial basis.
		\EndFor\label{step:pre-comp-end}
		
		\For{$\ell=1,\dotsc,n$}
		
			\For{$\vec{k}=(j_1,\dotsc,j_{\ell-1},i_{\ell+1},\dotsc,i_n)\in\C_{s,n-1}$}
			
				\Statew{3}{If $\R_{d-\left|\vec{k}\right|,s-\left|\vec{k}\bdiv{q}\right|,1}$ is nonempty (which fails to hold if and only if $d=sq-1$ and $\vec{k}=\vec{s}q$ for some $\vec{s}\in S_{s,n-1}$), then call Algorithm~\ref{alg:uni_hermite_eval_R} on the integers $d'=d-\left|\vec{k}\right|$ and $s'=s-\left|\vec{k}\bdiv{q}\right|$, the vector $(f_{(j_1,\dotsc,j_{\ell-1},i_\ell,i_{\ell+1},\dotsc,i_n)})_{i_\ell\in\R_{d',s',1}}$, and the polynomial $U_r\bmod{X^{s'-r}}$ for $r=\max(\floor{(d'+1)/q},0)$.}\label{step:eval_R}
			
			\EndFor
		
		\EndFor
	\end{algorithmic}
\end{algorithm}

We streamline notation during the complexity analysis of Algorithm~\ref{alg:multi_hermite_eval_R} by defining $\Delta_{d,s}=\left(s-\max(\floor{(d+1)/q},0)\right)q$ for $d,s\in\Z$, $\mult^*(0)=0$ and $\mult^*(k)=\mult(k)\log k$ for nonzero $k\in\N$. Then Lemma~\ref{lem:uni_hermite_eval_R} implies that Algorithm~\ref{alg:uni_hermite_eval_R} performs $\bigO(\mult^*(\Delta_{d-\left|\vec{k}\right|,s-\left|\vec{k}\bdiv{q}\right|}))$ operations in $\F_q$ during Line~\ref{step:eval_R} of Algorithm~\ref{alg:multi_hermite_eval_R}. The following lemma is used to bound the total number of operations performed by Algorithm~\ref{alg:uni_hermite_eval_R} over each iteration of the main loop.

\begin{lemma}\label{lem:awful_bnd} For $n\geq 2$ and $d,s\in\Z$ such that $d<sq$ and $s>0$,
\begin{equation}\label{eqn:awful_bnd}
	\sum_{\vec{k}\in\C_{s,n-1}}
	\Delta_{d-\left|\vec{k}\right|,s-\left|\vec{k}\bdiv{q}\right|}
	\leq\left(1+\frac{\max(\floor{d/q}+1,0)}{s}\right)\left|\R_{d,s,n}\right|.
\end{equation}
\end{lemma}
\begin{proof} We prove the lemma by induction on $n$. Suppose that $d,s\in\Z$ such that $d<sq$ and $s>0$. Then for all integers $\ell\geq 2$, we have
\begin{equation}\label{eqn:R_recursion}
	\R_{d,s,\ell}
	=\bigcup_{i_1\in\C_{s,1}}
		\left\{
		(i_1,i_2,\dotsc,i_\ell)
		\mid
		(i_2,\dotsc,i_\ell)\in\R_{d-i_1,s-(i_1\bdiv{q}),\ell-1}
	\right\}.
\end{equation}
It follows that
\begin{align*}\label{eqn:delta_bnd}
	\sum_{k\in\C_{s,1}}\Delta_{d-k,s-(k\bdiv{q})}
	&=\sum_{k\in\C_{s,1}}
	\left|\R_{d-k,s-(k \bdiv{q}),1}\right|
	+\left(\max\left(d-k+1,0\right)\bmod{q}\right)\\
	&=\left|\R_{d,s,2}\right|
	+\sum^d_{k=0}\left(d-k+1\bmod{q}\right).
\end{align*}
Therefore, the lemma is true for $n=2$ since
\begin{align*}
	\sum^d_{k=0}\left(d-k+1\bmod{q}\right)
	&\leq\max\left(\floor{\frac{d}{q}}+1,0\right)
	\sum^{q-1}_{k=0}\left(d-k+1\bmod{q}\right)\\
	&=\max\left(\floor{\frac{d}{q}}+1,0\right)\binom{q}{2}
\end{align*}
and
\begin{equation*}
	\left|\R_{d,s,2}\right|
	\geq\left|\R_{sq-1,s,2}\right|
	=\binom{s+1}{2}q^2-\binom{sq+1}{2}
	=s\binom{q}{2}.
\end{equation*}

Suppose now that $n\geq 3$ and that the lemma is true for all smaller values of $n$. Let $d,s\in\Z$ such that $d<sq$ and $s>0$. Then
\begin{equation}\label{eqn:awful_bnd_lhs}
	\sum_{\vec{k}\in\C_{s,n-1}}
	\Delta_{d-\left|\vec{k}\right|,s-\left|\vec{k}\bdiv{q}\right|}
	=\sum_{k\in\C_{s,1}}
	\sum_{\vec{k}\in\C_{s-(k\bdiv{q}),n-2}}
	\Delta_{d-k-\left|\vec{k}\right|,s-(k\bdiv{q})-\left|\vec{k}\bdiv{q}\right|}.
\end{equation}
For $k\in\C_{s,1}$, we have $d-k\leq d-(k\bdiv{q})q<(s-(k\bdiv{q}))q$ and $s-(k\bdiv{q})>0$. Thus, the induction hypothesis implies that for each $k\in\C_{s,1}$, the inner sum on the right-hand side of~\eqref{eqn:awful_bnd_lhs} is at most
\begin{equation*}
	\left(1+\frac{\max(\floor{(d-k)/q}+1,0)}{s-\floor{k/q}}\right)
	\left|
		\R_{d-k,s-(k\bdiv{q}),n-1}
	\right|.
\end{equation*}
Here, the first factor is always less than or equal to $1+\max(\floor{d/q}+1,0)/s$. Therefore, combining and substituting these upper bounds into~\eqref{eqn:awful_bnd_lhs} yields the inequality
\begin{equation*}
	\sum_{\vec{k}\in\C_{s,n-1}}
	\Delta_{d-\left|\vec{k}\right|,s-\left|\vec{k}\bdiv{q}\right|}
	\leq
	\left(1+\frac{\max(\floor{d/q}+1,0)}{s}\right)
	\sum_{k\in\C_{s,1}}
	\left|
		\R_{d-k,s-(k\bdiv{q}),n-1}
	\right|.
\end{equation*}
Equation~\eqref{eqn:R_recursion} with $\ell=n$ implies that the sum on the right-hand side of this inequality is equal to $\left|\R_{d,s,n}\right|$. Hence, \eqref{eqn:awful_bnd} holds and the lemma follows by induction.
\end{proof}

\begin{lemma}\label{lem:multi_hermite_eval_R} Algorithm~\ref{alg:multi_hermite_eval_R} performs
\begin{equation}\label{eqn:multi_hermite_eval_R}
	\bigO\left(
		\left(1+\frac{\floor{d/q}+1}{s}\right)
		\frac{\mult(\delta)\log\delta}{\delta}
		\left|\R_{d,s,n}\right|n
		+\ceil{\frac{s-1}{q-1}}\floor{\frac{d+1}{q}}
	\right)
\end{equation}
operations in $\F_q$, where $\delta=sq-\max\left(d+1-n(q-1),0\right)$.
\end{lemma}
\begin{proof} Each polynomial $U_r$ is of the form $V(X^{q-1})$ for some polynomial $V\in\F_q[X]$ of degree less than $\ceil{(s-r)/(q-1)}$. Thus, Lines~\ref{step:pre-comp}--\ref{step:pre-comp-end} of the algorithm perform at most $\ceil{(s-1)/(q-1)}\floor{(d+1)/q}$ operations in $\F_q$. As the polynomials $U_r$ are written on the monomial basis, no operations in $\F_q$ are performed in order to compute their residues in Line~\ref{step:eval_R}. Consequently, Lemma~\ref{lem:uni_hermite_eval_R} implies that Line~\ref{step:eval_R} performs $\bigO(\mult^*(\Delta_{d-\left|\vec{k}\right|,s-\left|\vec{k}\bdiv{q}\right|}))$ operations in $\F_q$. Hence, Algorithm~\ref{alg:multi_hermite_eval_R} performs
\begin{equation}\label{eqn:multi_R_bnd}
	\bigO\left(
		\ceil{\frac{s-1}{q-1}}\floor{\frac{d+1}{q}}
		+n\sum_{\vec{k}\in\C_{s,n-1}}
		\mult^{*\!}\left(
			\Delta_{d-\left|\vec{k}\right|,s-\left|\vec{k}\bdiv{q}\right|}
		\right)
	\right)
\end{equation}
operations in $\F_q$. As $\mult(k)/k$ is a nondecreasing function of $k$, so too is $\mult^*(k)/k$. Moreover, for $\vec{k}\in\C_{s,n-1}$,
\begin{align*}
	\Delta_{d-\left|\vec{k}\right|,s-\left|\vec{k}\bdiv{q}\right|}
	&=\left(s-\left|\vec{k}\bdiv{q}\right|\right)q
	-\max\left(
		\floor{(d-\left|\vec{k}\right|+1)/q}q,0
	\right)\\
	&\leq\left(s-\left|\vec{k}\bdiv{q}\right|\right)q
	-\max\left(d-\left|\vec{k}\right|+1-(q-1),0\right)\\
	&\leq sq
	-\max\left(d-\left|\vec{k}\bmod{q}\right|+1-(q-1),0\right)\\
	&\leq sq-\max\left(d+1-n(q-1),0\right).
\end{align*}
Letting $\delta=sq-\max\left(d+1-n(q-1),0\right)$, it follows that
\begin{equation*}
	\sum_{\vec{k}\in\C_{s,n-1}}
	\mult^{*\!}\left(
		\Delta_{d-\left|\vec{k}\right|,s-\left|\vec{k}\bdiv{q}\right|}
	\right)
	\leq
	\frac{\mult^{*\!}(\delta)}{\delta}
	\sum_{\vec{k}\in\C_{s,n-1}}
	\Delta_{d-\left|\vec{k}\right|,s-\left|\vec{k}\bdiv{q}\right|}.
\end{equation*}
Combining this inequality with~\eqref{eqn:awful_bnd} and~\eqref{eqn:multi_R_bnd} completes the proof.
\end{proof}

The factor $1+(\floor{d/q}+1)/s$ of the first term of the complexity bound~\eqref{eqn:multi_hermite_eval_R} measures the penalty that results from the complexity of Algorithm~\ref{alg:uni_hermite_eval_R} being a function of $\Delta_{d,s}$ rather than a function of $\left|\R_{d,s,1}\right|$. The former may be larger by a factor of $q-1$, while the penalty incurred by Algorithm~\ref{alg:multi_hermite_eval_R} is limited to a factor of~$2$. We have made no attempt to optimise the size of this factor, which would require strengthening the bound of Lemma~\ref{lem:awful_bnd}. For $d,s\in\N$ such that $d<sq$, $\R_{d,s,n}$ contains the vectors $\vec{j}+\vec{t}q$ for all $(\vec{j},\vec{t})\in[q]^n\times\N^n$ such that $\left|\vec{j}\right|\geq q$ and $\left|\vec{t}\right|=s-1$. Thus, we obtain the crude lower bound
\begin{equation*}
	\left|\R_{d,s,n}\right|
	\geq\binom{n-1+s-1}{n-1}\left(q^n-\binom{n+q-1}{n}\right).
\end{equation*}
From this bound, it is readily deduce that the value $\delta$ defined in Lemma~\ref{lem:multi_hermite_eval_R} is $\bigO(\left|\R_{d,s,n}\right|)$ for $n\geq 2$. Similarly, the second term of~\eqref{eqn:multi_hermite_eval_R} is $\bigO(\left|\R_{d,s,n}\right|)$ for $n\geq 3$. Thus, if $\mult(k)$ is taken to be in $\bigO(k\log k\log\log k)$, then Algorithm~\ref{alg:multi_hermite_eval_R} performs $\tilde{\bigO}(\left|\R_{d,s,n}\right|)$ operations in $\F_q$ for $n\geq 3$. For $n=2$, the second term of~\eqref{eqn:multi_hermite_eval_R} is only guaranteed to be $\bigO(\left|\C_{s,2}\right|)$, which is all that is required for fast encoding. Recall that the second term of~\eqref{eqn:multi_hermite_eval_R} counts the cost of Lines~\ref{step:pre-comp}--\ref{step:pre-comp-end} of Algorithm~\ref{alg:multi_hermite_eval_R}, which may be performed as a precomputation in many settings. With this precomputation and fast polynomial arithmetic, Algorithm~\ref{alg:multi_hermite_eval_R} attains quasi-linear complexity for $n=2$.

\subsection{Encoding algorithm}

The systematic encoding algorithm for high-rate multiplicity codes is presented in Algorithm~\ref{alg:high_rate}. The algorithm follows the approach described in Section~\ref{sec:high_rate}: first, the extended interpolation problem is solved in order to recover the polynomial $F_\C$,  after which $F_\R$ is deduced and used to compute the non-message entries of the encoding. Taking the univariate algorithm used by Algorithm~\ref{alg:multi_hermite_interp} to be Chin's interpolation algorithm, as modified in Section~\ref{sec:univariate}, it follows from Theorem~\ref{thm:multivariate_hermite} that Line~\ref{step:F_C} of the algorithm performs
\begin{equation*}
	\bigO\left(
		\frac{\mult(sq)\log sq}{sq}\left|\C_{s,n}\right|n
	\right)
\end{equation*}
operations in $\F_q$. Line~\ref{step:F_R} then performs $\left|\R_{d,s,n}\right|$ (cheap) multiplications by $-1$. Lines~\ref{step:uni_precomp} and~\ref{step:uni_eval_R} perform $\bigO(\mult(sq)\log sq)$ operations in $\F_q$ if $n=1$ (see Section~\ref{sec:uni_hermite_eval_R}), while Lemma~\ref{lem:multi_hermite_eval_R} bounds the number of operations performed by Line~\ref{step:multi_eval_R} if $n\geq 2$. Combining the bounds provides a second proof of Theorem~\ref{thm:encoders}.

\begin{algorithm}[!ht]
	\caption{Systematic encoding for high-rate codes}\label{alg:high_rate}
	\begin{algorithmic}[1]
		\Require Nonnegative integers $d$ and $s$ such that $d<sq$; and a vector $(c_{\vec{j}})_{\vec{j}\in\C_{s,n}}$ such that $m=(c_{\vec{j}})_{\vec{j}\in\I_{d,n}}$ is a message of $\Mult^s_d$, and $c_{\vec{j}}=0$ for $\vec{j}\in\R_{d,s,n}$.
		
		\Ensure The vector $(c_{\vec{j}})_{\vec{j}\in\C_{s,n}}$ such that
		\begin{equation*}
			\left(c_{\vec{i}}\right)_{\vec{i}\in\I_{d,n}}
			=F\dashv\left(N_{\vec{i}}\right)_{\vec{i}\in\I_{d,n}}
			\quad\text{and}\quad
			\left(c_{\vec{j}}\right)_{\vec{j}\in\R_{d,s,n}}
			=\left(\evH\left(F,\vec{j}\right)\right)_{\vec{j}\in\R_{d,s,n}}
		\end{equation*}
		for the polynomial $F\in\F_q[\vec{X}]_d$ that corresponds to the message $m$.
		
		\State\label{step:F_C}Call Algorithm~\ref{alg:multi_hermite_interp} on $\C_{s,n}$ and $(c_{\vec{j}})_{\vec{j}\in\C_{s,n}}$.
		
		\State\label{step:F_R}Set $c_{\vec{j}}$ equal to $-c_{\vec{j}}$ for each $\vec{j}\in\R_{d,s,n}$.
		
		\If{$n=1$}
		
			\State\label{step:uni_precomp} Compute $U=(X^{q-1}-1)^r\bmod{X^{s-r}}$ for $r=\floor{(d+1)/q}$.
			
			\State\label{step:uni_eval_R} Call Algorithm~\ref{alg:uni_hermite_eval_R} on $d$, $s$, $(c_j)_{j\in\R_{d,s,1}}$ and $U$.
			
		\Else
		
			\State\label{step:multi_eval_R} Call Algorithm~\ref{alg:multi_hermite_eval_R} on $d$, $s$ and $(c_{\vec{j}})_{\vec{j}\in\R_{d,s,n}}$.
			
		\EndIf
	\end{algorithmic}
\end{algorithm}

\section{Conclusion}

We presented two quasi-linear time systematic encoding algorithms for multiplicity codes which provide complimentary performance in practical settings. Of the two algorithms, the one that provides the shortest encoding time for a given set of parameters will vary with the choice of univariate algorithms. Moreover, their encoding times depend on additional factors besides the number of field operations they perform. Consequently, we cannot draw solid conclusions about the relative performance of the two algorithms in a given practical setting by comparing their stated complexities (which would also require estimating hidden constants). However, as the encoding algorithms share the same underlying subroutines, implementing both algorithms on a particular architecture should not require much more effort than implementing just one of the algorithms, and would allow for direct comparisons to be made.

In Section~\ref{sec:improvements}, we described modifications to the encoding algorithm for low-rate codes which were aimed at improving its practical performance by eliminating some unnecessary operations. Similar modifications may also be made to the encoding algorithm for high-rate codes, but we omit details here. The algorithm for high-rate codes would also benefit from improving or replacing Algorithm~\ref{alg:uni_hermite_eval_R} so that Line~\ref{step:eval_R} of Algorithm~\ref{alg:multi_hermite_eval_R} can always be performed in $\tilde{\bigO}(\left|\R_{d-\left|\vec{k}\right|,s-\left|\vec{k}\bdiv{q}\right|,1}\right|)$ operations and without the need to provide the additional polynomial input.

The Hermite interpolation and evaluation algorithms of Section~\ref{sec:hermite} may be of independent interest. The dependency of the multivariate algorithms on the univariate algorithms draws our attention to the problem of optimising the choice of univariate algorithms. In this direction, it would be interesting to develop fast univariate Hermite interpolation and evaluation algorithms that work natively on the Newton basis, which would allow us to avoid the basis conversions that are necessary when using the algorithms discussed in Section~\ref{sec:univariate}. Encouragingly, such algorithms already exist for instances that do not involve derivatives~\cite[Section~5.1]{bostan2005}. It would also be interesting to further investigate the benefits offered by the techniques of van der Hoeven~\cite{hoeven2016} when using the algorithms of Section~\ref{sec:univariate}. 

\section*{Acknowledgements}

The author would like to thank Daniel Augot and Fran\c{c}oise Levy-dit-Vehel for their helpful comments on this and earlier versions of the paper, and Joris van der Hoeven for bringing reference~\cite{hoeven2016} to the author's attention.

\bibliographystyle{amsplain}

\end{document}